\documentclass[conference]{IEEEtran}
\IEEEoverridecommandlockouts

\usepackage[utf8]{inputenc}
\usepackage[T1]{fontenc}

\usepackage{cite}
\usepackage{subfig}
\usepackage{graphicx}
\usepackage[export]{adjustbox}

\usepackage{balance}

\usepackage[normalem]{ulem}

\usepackage{amsmath,amsfonts,amsthm}

\newtheorem{proposition}{Proposition}

\newcommand{\p}{\ensuremath{\mathcal{P}}}
\newcommand{\np}{\ensuremath{\mathcal{N\!P}}}

\newcommand{\fmax}{\ensuremath{F_\text{max}}{}}
\newcommand{\smax}{\ensuremath{S_\text{max}}{}}
\newcommand{\csum}{\ensuremath{\sum C_i}{}}
\newcommand{\ssum}{\ensuremath{\sum S_i}{}}
\newcommand{\cwsum}{\ensuremath{\sum w_iC_i}{}}

\DeclareMathOperator*{\argmin}{arg\,min}

\usepackage{algorithm}
\usepackage[noend]{algpseudocode}
\algdef{SE}[DOWHILE]{Do}{doWhile}{\algorithmicdo}[1]{\algorithmicwhile\ #1}%

\begin{document}

\title{Data-driven scheduling in serverless computing\\ to reduce response time
\thanks{Preprint of the paper accepted at the 21th IEEE/ACM International Symposium on Cluster, Cloud and Internet Computing (CCGrid 2021), Melbourne, Australia, 2021}
}

\author{
    \IEEEauthorblockN{%
        Bartłomiej Przybylski\IEEEauthorrefmark{1}, Paweł Żuk\IEEEauthorrefmark{2}, Krzysztof Rzadca, \emph{Member, IEEE}\IEEEauthorrefmark{3}}
    \IEEEauthorblockA{%
        \emph{Institute of Informatics, University of Warsaw}\\
        Warsaw, Poland\\
        Email: \IEEEauthorrefmark{1}bap@mimuw.edu.pl, \IEEEauthorrefmark{2}p.zuk@mimuw.edu.pl, \IEEEauthorrefmark{3}krzadca@mimuw.edu.pl}
}

\maketitle

\begin{abstract}
In Function as a Service (FaaS), a serverless computing variant, customers deploy functions instead of complete virtual machines or Linux containers. It is the cloud provider who maintains the runtime environment for these functions. FaaS products are offered by all major cloud providers (e.g. Amazon Lambda, Google Cloud Functions, Azure Functions); as well as standalone open-source software (e.g. Apache OpenWhisk) with their commercial variants (e.g. Adobe I/O Runtime or IBM Cloud Functions).
We take the bottom-up perspective of a single node in a FaaS cluster. We assume that all the execution environments for a set of functions assigned to this node have been already installed. Our goal is to schedule individual invocations of functions, passed by a load balancer, to minimize performance metrics related to response time. Deployed functions are usually executed repeatedly in response to multiple invocations made by end-users. Thus, our scheduling decisions are based on the information gathered locally: the recorded call frequencies and execution times.
We propose a number of heuristics, and we also adapt some theoretically-grounded ones like SEPT or SERPT. Our simulations use a recently-published Azure Functions Trace. We show that, compared to the baseline FIFO or round-robin, our data-driven scheduling decisions significantly improve the performance.
\end{abstract}


\section{Introduction}

Serverless computing \cite{Castro2019} allows cloud customers to execute their code without configuring and maintaining a production environment or a software infrastructure stack. Major cloud offer serverless products, e.g. Amazon Lambda, Google Cloud Functions, and Microsoft Azure Serverless. In this paper, we focus on a variant of serverless computing called \emph{Function as a Service (FaaS)} \cite{Fox2017}. In FaaS, a \emph{cloud customer} develops a source code of a stateless function and then uploads it to the cloud provider. When a function is invoked by an \emph{end-user}, this \emph{invocation} is processed on the infrastructure managed by the provider.

We consider a set of functions that have been already loaded into memory of a single node in a large cluster. We intentionally omit function-to-node assignment to show that the performance of the whole cluster can be improved on the node-level too. Such an improvement is orthogonal to improvements in function placement~\cite{kaffes2019centralized,suresh2020ensure}, load-balancing~\cite{suresh2020ensure} or auto-scaling of the clusters~\cite{Perez2019}.

In FaaS, each function can be invoked numerous times, e.g. in response to repeated HTTP requests coming from various end-users. Thus, the local (node) scheduler can make online decisions on how to assign these invocations to available CPU cores based on invocations from the past. The information used may include, among others, the frequency of invocations and their observed past execution times. For this reason, theoretical lower bounds for competitiveness of online strategies such as SPT or SRPT (see, e.g., \cite{Bender1998,Muthukrishnan1999,Megow2004}) can be too conservative in case of the analysed problem (formally introduced in Sec.~\ref{sec:pd}).

These local scheduling decisions can be made implicitly by the kernel scheduler (at the operating system level). However, the kernel has a low-level perspective on the scheduling problem. In particular, the kernel is not aware of how individual FaaS invocations map to processes (threads), so, in the context of FaaS, it cannot make dynamic, function-related decisions. As a consequence, the kernel is forced to use variants of the round-robin approach, where individual invocations are processed in turns and thus repeatedly preempted. We show that the overall performance of the system can be increased by implementing other reasonable heuristics on the local scheduler level with no significant computational or memory cost---as long as we can use the information about how the individual invocations link to the functions.

We state that in practice one rarely deals with extreme generality (as in the theoretical, worst-case results), and that better decisions can be made by taking into account information readily available on a local FaaS node. This is so independently of whether information is known a priori or it is guessed (predicted) based on historical data. We validate this claim with computational experiments using the real-life data recently published as the Azure Functions Trace \cite{Shahrad2020}.

The contributions of this paper are as follows:
\begin{itemize}
    \item We define a theoretical model of node-level scheduling for FaaS (Sec.~\ref{sec:pd}). We adapt the real-life data from the Azure Function Trace to reflect our model (Sec.~\ref{sec:azure-dataset}--\ref{sec:model}).
    \item We propose a number of theoretically-grounded heuristics and a new one, \emph{Fair Choice}, that can be used by the local scheduler to make decisions online. We also show that these heuristics can be implemented without significant increase of auxiliary computations (Sec.~\ref{sec:st}).
    \item By simulations, we show that applying heuristics based on past information leads to reduction in latency-related objectives, compared to the preemptive round-robin (corresponding to standard scheduling used by the operating system, Sec.~\ref{sec:ev}). 
\end{itemize}

\section{Problem description}
\label{sec:pd}

In this section, we define the optimization problem of minimizing latency-related objectives in the FaaS environment. The aim of this problem is to be simply-defined, yet realistic enough to address dilemmas encountered in the serverless practice. Our notation follows the standard of Brucker~\cite{Brucker2007}.

We consider a single physical machine with $m$ parallel processors/cores $P_1, P_2, \dots, P_m$ (a processor is a standard scheduling term; our processor maps to a single core on the machine). This machine has been assigned a set of $n$ stateless functions, $f_1, f_2, \dots, f_n$, that can be invoked multiple times, without a significant startup time (i.e. they are already loaded into memory). The functions are stateless, so one processor can execute one function at a time, but at any moment invocations of the same function can be independently processed on different processors. Each invocation (call) corresponds to a single end-user request. The actual execution time of $f_j$ differs between calls, and we model it as a random variable with the $\mathbb{P}_{j}$ distribution. We consider both a preemptive and a non-preemptive case. In the preemptive case, a process executing a function can be suspended by the operating system, and later restored on the same or on another processor. In the non-preemptive case, a process executing a function---once started---occupies the processor until the function finishes and the result is ready to be returned to the end-user.

We assume that the $f_1, f_2, \dots, f_n$ functions, assigned to the machine, have been selected by an external load balancer/supervisor. In the case of cloud clusters, where thousands of physical machines work simultaneously, this process may take into account complex placement policies (balancing the overall load, affinity to reduce cluster network load, anti-affinity to increase reliability by placing instances executing the same function on different nodes or racks). In this paper, we focus on the micro-scale of a single node in such a cluster.

We consider a time frame of $[0, T)$ where $T$ is a positive integer. Each function $f_j$ can be executed multiple times in response to numerous calls incoming in this time frame. Thus, the instance of the considered problem can be described as a sequence of an unknown number of invocations (calls) in time. Let the $i$-th call be represented by a pair of values: the moment of call $r(i)$ and a reference $f(i)$ to the invoked function. Of course, $0 \leq r(i - 1) \leq r(i) < T$ for all $i > 1$.

Once the $i$-th invocation is finished (e.g. the result is returned to the end-user), we know the moment of its completion $c(i)$, and the total processing time $p(i)$ that the invocation required. Note that it is possible that $c(i) \geq T$. We use two base metrics to measure the performance of handling a call. The flow time $F(i)$, i.e. $c(i) - r(i)$, corresponds to the server-side processing delay of the query. The stretch (also called the slowdown) $S(i)$, i.e. $F(i)/p(i)$, weights the flow time by the processing time.

For any instance of our problem, processes executing functions need to be continuously assigned to processors in response to incoming calls. We determine the quality of the obtained schedule based on the following performance metrics that aggregate flow times or stretches across all the calls.

\begin{itemize}
    \item \textbf{Average flow time} (AF), $\sum_i F(i)/\#\{i\},$ where $\#\{i\}$ is a total number of invocations. It is a standard performance metric considered for over four decades in various industrial applications \cite{Baker1974introduction}. It corresponds to the average response time.
    \item \textbf{Average stretch} (AS), $\sum_i S(i)/\#\{i\},$ which takes into account the actual execution time of a call \cite{Bender1998}, and thus responds to the observation that it is less noticeable that a 2-second call is delayed by 40 milliseconds (with the stretch of $2.04/2 = 1.02$) than it would be in case of a 10-millisecond call (resulting in the stretch of $50/10 = 5$).
    \item \textbf{99th percentile of flow time} (F99), $x \colon \mathbb{P}(F(i) < x) = 0.99, \text{and}$
    \item \textbf{99th percentile of stretch} (S99), $x \colon \mathbb{P}\left(S(i) < x\right) = 0.99,$ which are less fragile variants of the maximum performance metrics \cite{Bender1998}. We state that maximum-defined metrics are not good indicators of the performance in FaaS, as if the flow time or stretch of a call exceeds a perceptual threshold accepted by the end-user, the call is cancelled and the function is called again (e.g. by refreshing a webpage). This perceptual threshold reduces the number of calls with an unacceptably high flow time that have a significant impact on the overall performance \cite{Sloss2019}. Our robust variants, measuring the $99$th percentile, return a value $x$ such that 99\% of all invocations have stretch (or flow time) smaller than $x$.
    \item \textbf{Average function-aggregated flow time} (FF), $$\frac{1}{n}\sum_{j = 1}^{n} \frac{\sum_{\{i \colon f(i) = j\}}F(i)}{\# \{i \colon f(i) = j\}},$$ where $\# \{i \colon f(i) = j\}$ is the number of all $f_j$ calls, and
    \item \textbf{Average function-aggregated stretch} (FS), $$\frac{1}{n}\sum_{j=1}^{n} \frac{\sum_{\{i \colon f(i) = j\}}F(i)}{\sum_{\{i \colon f(i) = j\}} p(i)},$$ which we propose as new metrics specific for the FaaS environment. Our aim is to measure the fairness of a schedule, based on the average values of the flow time or stretch within the sets of invocations of the same functions. These metrics take into account that functions developed by different users may require different amounts of resources (i.e. time), and that the performance of an invocation should not depend significantly on the set of functions that share the same machine.
\end{itemize}

As the actual processing times of invocations are random variables, we solve a set of online stochastic scheduling problems. Using the extended three-field notation \cite{Brucker2007}, we denote these problems as $\text{P}m|\text{on-line}, r(i), p(i) \sim \mathbb{P}_{f(i)}|\mathbb{E}[\sigma]$ and $\text{P}m|\text{on-line}, \text{pmtn}, r(i), p(i) \sim \mathbb{P}_{f(i)}|\mathbb{E}[\sigma]$ where $\sigma \in \{\text{AF, AS, F99, S99, FF, FS}\}$.

\begin{proposition}
The $\text{P}m|\text{pmtn}, r(i)|\text{FF}$ and $1|r(i)|\text{FF}$ problems are strongly \np-Hard. The $1|r(i)|\text{FS}$ problem is \np-Hard.
\end{proposition}

\begin{proof}
Consider special cases of the above problems, in which each function is invoked exactly once. Then, the FF metric becomes equivalent to the $\sum_i c(i)$ metric. The $\text{P}m|\text{pmtn}, r(i)|\sum_i c(i)$ \cite{Bellenguez2015} and $1|r(i)|\sum_i c(i)$ \cite{LenstraRinnooy-Kan1977} problems are strongly \np-Hard. Similarly, the FS metric becomes equivalent to the $\sum_i S(i)$ metric. It is known that the $1|r(i)|\sum_i S(i)$ problem is \np-Hard \cite{Legrand2006}.
\end{proof}

\section{Measuring invocations in the Azure dataset}
\label{sec:azure-dataset}

The recently published Azure Function Trace~\cite{Shahrad2020} provides information about function invocations collected over a continuous 14-day period between July 15th and July 28th, 2019.
For each day within this period, the trace presents a number of invocations of each of the monitored functions during each minute of the day ($24 \cdot 60 = 1440$ separate measurements).
We denote the number of invocations of the $f_j$ function within the $k$-th minute of the trace as $\lambda^k_j$.
The trace distinguishes between different invocation sources, e.g. incoming HTTP requests or periodic executions (\texttt{cron} tasks). In this paper, we consider HTTP requests only, as they are less predictable.
Additionally, the trace shows the distribution of the execution times for each function during each day (based on weighted averages from 30-second intervals).
For each function, the trace provides values of the $0$th, $1$st, $25$th, $50$th, $75$th, $99$th and the $100$th percentile of this approximate distribution.

\begin{figure}
    \includegraphics[width=\columnwidth]{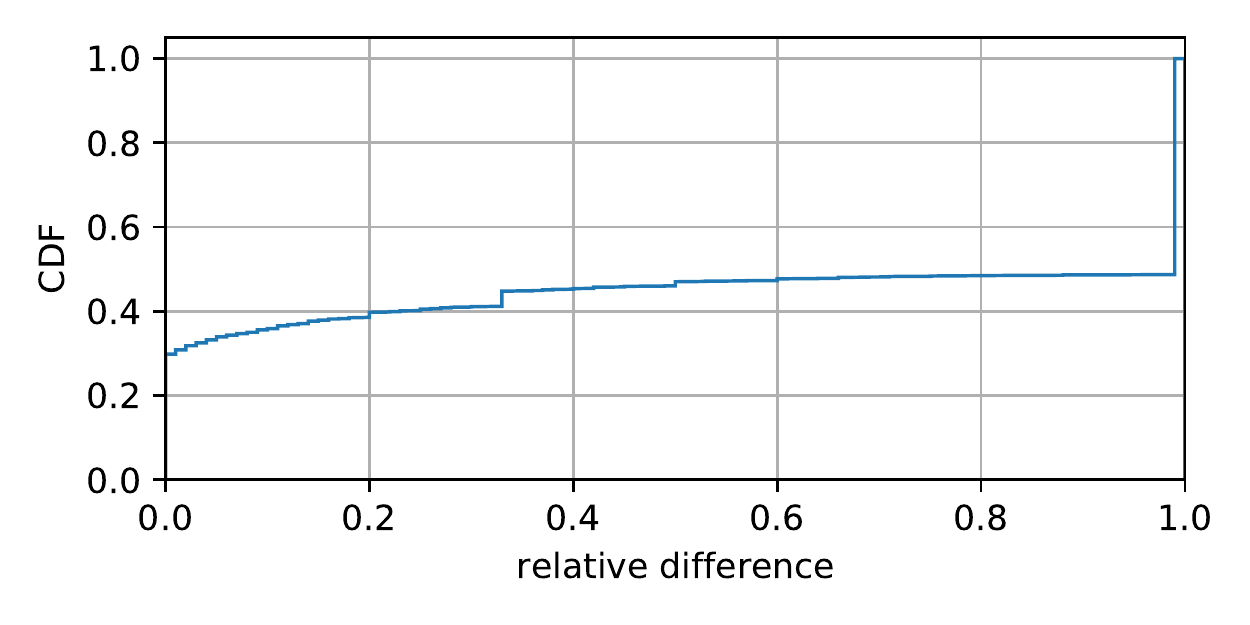}
    \caption{CDF of the relative difference $\Delta^k_j$ of the number of invocations compared to the previous one-minute interval over $10\,000$ randomly sampled interval-function pairs $(k, j)$, where at least one of the $\lambda^k_j$ and $\lambda^{k-1}_j$ values is non-zero.
    }
    \label{fig:lambda_diff_cdf}
\end{figure}
\begin{figure}
    \includegraphics[width=\columnwidth]{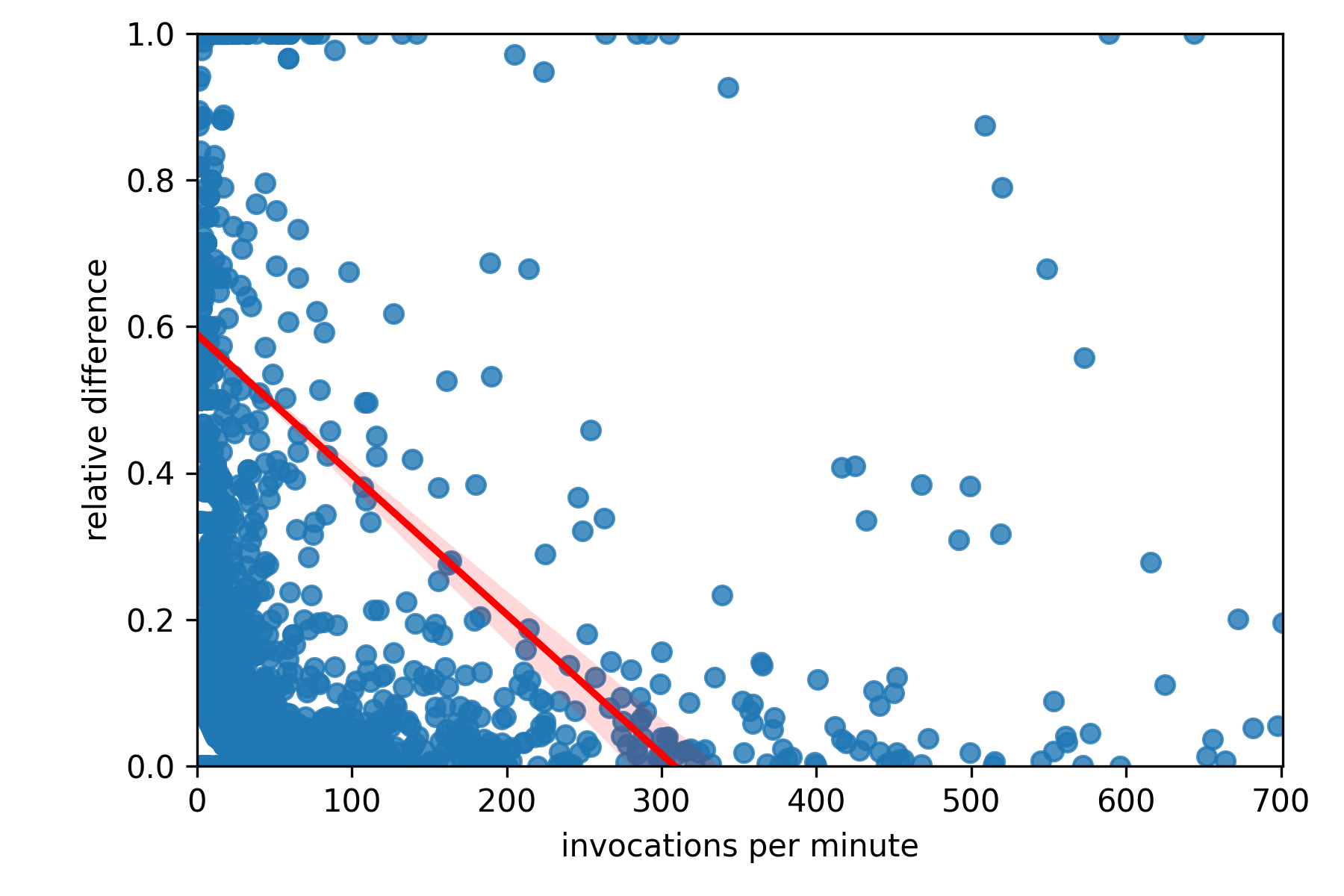}
    \caption{Relative differences for numbers of invocations compared to the previous one-minute interval. Each point shows one of $10\,000$ randomly sampled interval-function pairs $(k, j)$, where at least one of the $\lambda^k_j$ and $\lambda^{k-1}_j$ values is non-zero.
    Red line represents linear regression fit to the visible data.
    Bands around the line indicate $90\%$ confidence interval.
    }
    \label{fig:lambda_pred_cmp}
\end{figure}

Some of the scheduling algorithms presented in Sec.~\ref{sec:st} use the expected number of invocations in the $k$-th interval, $\lambda^k_j$, to make decisions online. The most straightforward way to estimate the unknown $\lambda^k_j$ is to use the number of invocations in the previous interval, $\lambda^{k-1}_j$. We analyzed how the actual number of invocations differed between two consecutive one-minute intervals. In particular, we define the \emph{relative difference} between $\lambda^k_j$ and $\lambda^{k-1}_j$ as
\[ \Delta^k_j = \begin{cases}
    0  & \text{if $\lambda^k_j = \lambda^{k-1}_j = 0$ or $\lambda^{k-1}_j$ not known} \\
    \frac{\left| \lambda^k_j  - \lambda^{k-1}_j  \right|}{\lambda^k_j + \lambda^{k-1}_j} & \text{otherwise}
\end{cases} \]

We calculated relative differences $\Delta^k_j$ for all the recorded functions $f_j$ and minutes $k$. 
We found out that $85\%$ of these were equal to zero (meaning that the number of invocations did not change): for $93\%$ of these cases (and 79\% of all) there were no invocations in both the minutes ($\lambda^k_j = \lambda^{k-1}_j = 0$).
Next, we studied in detail the cases for which $\lambda^k_j + \lambda^{k-1}_j > 0$. 
We randomly sampled $10\,000$ pairs $(k, j)$ for which $\lambda^k_j + \lambda^{k-1}_j > 0$, and calculated the corresponding relative differences.
Fig.~\ref{fig:lambda_diff_cdf} shows the CDF of the obtained $\Delta^k_j$ values. In $29\%$ of the cases the number of invocations did not change; but for roughly $51\%$ the relative difference was $1.0$, denoting cases in which either $\lambda^k_j$ or $\lambda^{k-1}_j$ was 0. Fig.~\ref{fig:lambda_pred_cmp} shows $\Delta^k_j$ as a function of $\lambda^k_j$ (for clarity of the presentation, the x-axis is cut at the 99th percentile of all $\lambda^k_j$ values). We see that the relative difference decreases with the increased number of invocations per minute.

In further sections we present algorithms taking advantage of this knowledge to predict number of invocations to come. However, our algorithms additionally rely on the $p(i)$ values, the execution times of invocations. Due to the limited space, for the analysis of $p(i)$ in the trace we refer the reader to Sec.~3.4 of the excellent analysis in~\cite{Shahrad2020}.

\section{Mapping the Azure dataset to our model}
\label{sec:model}

The theoretical model introduced in Sec.~\ref{sec:pd} is an on-line scheduling problem with release times and stochastic processing times. In this section, we show that monitoring data from Azure, a real-world FaaS system, is sufficient to fulfill the assumptions we take in the model (e.g., that the processing times are generated by an arbitrary distribution). The main issue is that when the number of events (e.g. function invocations) is large, it can be monitored only in aggregation---as in the case of the number of invocations and function execution times in the Azure dataset. Thus, we need to extrapolate these aggregations. Below, we describe how we acquire two sets of parameters required by our theoretical model: invocation times and the random distribution of processing times.

\subsection{Invocation times}

The Azure dataset does not give us the exact moments of invocation of each function. However, the total \emph{number} $\lambda^k_j$ of invocations of each function $f_j$ is known for every $k$-th monitored minute. Thus, we assume that $T$, the duration of the considered time frame, is a multiple of $60\,000$ (number of milliseconds in a minute; our base unit is a millisecond because the processing times are given in milliseconds). The number of invocation, $\lambda^k_j$, may change in the $[0, T)$ interval as, for example, some functions are called intensively in the morning and rarely during the night. In order to model such changes, we divide the $[0, T)$ time frame into $K = T/60\,000$ consecutive, one-minute intervals $v^k$: $$\begin{array}{l}
    v^1 = [0, 60\,000),\\
    v^2 = [60\,000, 120\,000),\\
    \dots\\
    v^K = [(K-1)\cdot 60\,000, T).\\
\end{array}$$
We use the values of $\lambda_j^k$ obtained directly from the Azure dataset to generate invocation times of function $f_j$ in these intervals. Following a standard queueing theory, we assume that in the interval $v^k$, each function $f_j$ is called based on the Poisson point process with rate $\lambda^k_j/60\,000$. Thus, the time (in milliseconds) between consecutive calls of function $f_j$ in the $v^k$ interval is a random variable with the $\text{Exp}(\lambda_j^k/60\,000)$ distribution. Our scheduling model does not rely on this assumption---we use $r(i)$ values, the realizations of the random variables. However, some of our scheduling algorithms estimate $\lambda^k_j$ for better scheduling decisions.

\subsection{Processing times}

Precise information on $p(i)$ values, the execution times of single invocations, is not provided in the Azure dataset. Our model assumes the existence of the distribution $\mathbb{P}_j$ (of execution times of a function $f_j$) in its exact form. However, the Azure dataset only shows selected percentiles of the empirical cumulative distribution function ($p_0$, percentile $0$, $p_1$, the $1$st percentile, $p_{25}$, $p_{50}$, $p_{75}$, and $p_{100}$). We thus approximate $\mathbb{F}_j$, the CDF of $\mathbb{P}_j$, by a piecewise-linear interpolation of these percentiles. For example, if $x \in (p_{1}, p_{25}]$, then $$\mathbb{F}_j(x) = 0.01 + 0.24\cdot\frac{x - p_{1}}{(p_{25} - p_{1})}.$$
The actual processing time of each invocation is generated based on the approximation of $\mathbb{F}_j$.

\section{Scheduling algorithms}
\label{sec:st}

Within the mapping described in the previous section, we apply both the well-known and theoretically-grounded strategies, and new approaches. In particular, we consider the following strategies.
\begin{itemize}
    \item \textbf{FIFO} (\emph{First In, First Out}, for a non-preemptive case) --- all invocations are queued in the order in which they were received. When a processor is available, it is assigned the invocation with the lowest value of $r(i)$.
    \item \textbf{SEPT} (\emph{Shortest Expected Processing Time}, for a non-preemptive case) --- when a processor is available, it is assigned the invocation with the shortest expected processing time.
    \item \textbf{FC\#} (\emph{Fair Choice} based on the number of invocations, for a preemptive and a non-preemptive case) --- when a processor is available, it is assigned the invocation which is the most unexpected. In fact, we want functions that are called occasionally to have larger priority than the frequently-invoked ones. At any point $t \in v^k$, the most unexpected invocation is the one with the lowest value of $\max\left\{\lambda_j^k, \#\{ i \colon r(i) \in v^k \text{ and } f(i) = j \}\right\}$. The priority is thus determined based on the maximum of the expected number of invocations in the considered period and the actual number of these invocations.
    \item \textbf{FCP} (\emph{Fair Choice} based on the total processing time, for a preemptive and a non-preemptive case) --- when a processor is available, it is assigned the invocation related to the least demanding function in total. In fact, we want functions that use limited resources to have larger priority than the burdening ones. At any point $t \in v^k$, the least demanding invocation is the one with the lowest value of $$\max\left\{\lambda_j^k\cdot\mathbb{E}[X \sim \mathbb{P}_j], {\sum_{\{i \colon r(i) \in v^k \text{ and } f(i) = f_j\}} p(i)}\right\}.$$ The priority is thus determined based on the maximum of the expected total processing times of invocations in the considered one-minute interval and the actual value of this amount.
    \item \textbf{RR} (\emph{Round-robin}, for a preemptive case) --- all invocations that have not been completed are queued. When a processor is available, it is assigned the first invocation from the queue. If the execution does not complete by a fixed period of time is it preempted and moved to the end of the queue. We consider periods of the length of $10$, $100$ and $1000$ milliseconds. As this strategy is used as a reference, we assume that all the invocations have the same priority.
    \item \textbf{SERPT} (\emph{Shortest Expected Remaining Processing Time}, for a preemptive case) --- when a processor is available, it is assigned the invocation with the shortest expected remaining processing time.
    The strategy is applied only when an invocation is finished or a new call is received, even if all the processors are busy. Such a restriction is introduced because otherwise executions of functions with increasing expected remaining processing times could be preempted an arbitrarily large number of times.
\end{itemize}

In all cases, if the rule is unambiguous, we select the invocation with the lowest value of $r(i)$. The expected (remaining) processing times and the values of $\lambda_j^k$ are estimated based on previous invocations.

Although the FIFO, SEPT, RR and SERPT strategies are well-known, we implemented them among our solutions. For each of SEPT, SERPT, FC\# and FCP strategies, we introduce two methods: \textsc{Position} and \textsc{Update}. The first method returns the position of an invocation in the execution queue. For example, in case of SEPT, it returns the expected processing time of a function. The \textsc{Update} method is called after the execution of the invocation ends, and it updates auxiliary data structures that are used by the \textsc{Position} method. The framework algorithm (Alg.~\ref{alg:framework}) for the above strategies is based on a standard scheduling loop. In this loop, we prioritize calls and choose the one with the lowest position.

\begin{algorithm}[tb]
    \caption{A framework scheduling algorithm}
    \label{alg:framework}
    \begin{algorithmic}[1]
    \State {$Q \gets \{\}$} \Comment{A queue of incoming invocations}
    \State {$A \gets \{\}$} \Comment{A set of acknowledged invocations}
    \State {$E \gets \{\}$} \Comment{A set of invocations being processed}
    \State {$\text{preemptive} \in \{\text{true}, \text{false}\}$} \Comment {Is the strategy preemptive?}
    \While{true}
        \State {\textbf{wait until} (a processor is free \textbf{and} ($|A| + |Q| > 0$)) \textbackslash \newline \hspace*{4em}\textbf{or} (preemptive \textbf{and} $|Q| > 0$)}
        \For{\textbf{each} $i$-th call in $E$ that has been finished}
            \State $\textsc{Update}(f(i), p(i))$
            \State {Remove $i$ from $E$}
        \EndFor
        \State {Move all the invocations from $Q$ to $A$}
        \If{preemptive}
            \State {Move all the invocations from $E$ to $A$}
        \EndIf
        \While {$|A| > 0$ \textbf{and} there are free processors}
            \State $i' \gets \argmin_{i \in A} \textsc{Position}(f(i), p(i))$
            \State \Comment Here, $p(i)$ is a partial processing time
            \State {Assign the $i'$-th call to any free processor}
            \State {Move $i'$ from $A$ to $E$}
        \EndWhile
    \EndWhile
    \end{algorithmic}
\end{algorithm}

The positions of invocations are calculated differently for different strategies. We present pseudocodes for two of them, a non-preemptive SEPT and a preemptive SERPT, as FC\# and FCP are similar. Alg.~\ref{alg:SEPT} defines the \textsc{Position} and \textsc{Update} methods for case of the SEPT strategy. As the invoked function is known at the moment of the call, the methods are similar for all the functions, but the auxiliary data structures are function-dependent. At any point in time, we store two values for each function $f_j$: the total processing time ($\text{TPT}_j$) of all its previous invocations and a number of these invocations ($\text{NOC}_j$). The expected processing time is approximated using a standard estimator, the average execution time $\text{TPT}_j / \text{NOC}_j$.

\begin{algorithm}[tb]
    \caption{SEPT}
    \label{alg:SEPT}
    \begin{algorithmic}[1]
    \For{$j \in \{1, 2, \dots, n\}$} \Comment{Set initial values}
        \State $\text{TPT}_j, \text{NOC}_j \gets 0, 0$
    \EndFor

    \Procedure{Position}{$j$, $p$} \Comment{$p$ is always $0$}
    \If{$\text{NOC}_j > 0$}
        \Return $\text{TPT}_j/\text{NOC}_j$
    \ElsIf{$\sum_j\text{NOC}_j = 0$}
        \Return $0$
    \Else{}
        \Return $\sum_j\text{TPT}_j/\sum_j\text{NOC}_j$
    \EndIf
    \EndProcedure
    
    \Procedure{Update}{$j, p$}
        \State $\text{TPT}_j \gets \text{TPT}_j + p$
        \State $\text{NOC}_j \gets \text{NOC}_j + 1$
    \EndProcedure
    \end{algorithmic}
\end{algorithm}

Similarly, Alg.~\ref{alg:SERPT} provides the same methods for the SERPT strategy. For each function $f_j$, we store execution times of its previous invocations in the $\text{PT}_j$ vector. In general, this vector can be arbitrarily long. However, to reduce the memory footprint of the algorithm, we might want to limit its size. In such a case, the oldest values can be replaced with the newest ones. This approach has two main advantages: (1) if the distribution of processing times changes in time, it can be reflected within the algorithm, (2) the memory is saved. On the other hand, if the number of remembered values is limited, the accuracy of the estimation of the expected remaining processing time (ERPT) is limited too. For clarity of the presentation, we omit most of the implementation details. For example, one can use binary search or priority queues to improve the complexity of the presented approach.

In order to determine the ERPT of a call of $f_j$ after $p$ milliseconds, we select from $\text{PT}_j$ the execution times that were equal to or exceeded $p$. We then estimate the expected processing time of the current invocation based on the standard estimator, similar to the one presented in case of SEPT.

\begin{algorithm}[tb]
    \caption{SERPT}
    \label{alg:SERPT}
    \begin{algorithmic}[1]
    \For{$j \in \{1, 2, \dots, n\}$} \Comment{Set initial values}
        \State $\text{PT}_j \gets \{\}$ \Comment{A set of previous processing times}
    \EndFor

    \Procedure{Position}{$j$, $p$}
    \State {$pp, pc \gets 0, 0$}
    \For {$pt \in \{x \in \text{PT}_j \colon x - p \geq 0\}$}
        \State {$pp \gets pp + (pt - p)$}
        \State {$pc \gets pc + 1$}
    \EndFor

    \If{$pc > 0$}
    \Return $pp/pc$
    \Else{}
        \For{$pt \in \{x \in \cup_j\text{PT}_j \colon x - p \geq 0\}$}
            \State {$pp \gets pp + (pt - p)$}
            \State {$pc \gets pc + 1$}
        \EndFor
        \If{$pc > 0$}
        \Return $pp/pc$
        \Else{}
        \Return $0$
        \EndIf
    \EndIf
    \EndProcedure
    
    \Procedure{Update}{$j, p$}
        \State {Add $p$ to $\text{PT}_j$}
    \EndProcedure
    \end{algorithmic}
\end{algorithm}

\section{Evaluation}
\label{sec:ev}

We evaluate and compare our algorithms using discrete-time simulations, as this method enables us to perform evaluations on a large scale. In this section, we present the results of these evaluations.
In Sec.~\ref{sec:ev-est}, we present estimation models used with the algorithms.
To make sure that our input data matches real-world scenarios, we generate test instances based on the Azure Functions Trace, with mapping indicated in Sec.~\ref{sec:model}.
In Sec.~\ref{sec:ev-preprocess}, we describe data preprocessing, and in Sec.~\ref{sec:ev-method} we describe how we create inputs for our simulator.
Finally, in Sec.~\ref{sec:ev-metrics}--\ref{sec:ev-params}, we analyse results of the simulations and the behavior of the tested algorithms.

We implemented the simulator and all the algorithms in \textsc{c++} and run the simulator on Intel Xeon Silver 4210R CPU @ 2.40GHz with 250GB RAM. The simulator was validated using unit tests and a close-up inspection of results on small instances.

\subsection{Estimations and baselines}
\label{sec:ev-est}
Our algorithms rely heavily on probabilistic estimations of parameters (e.g. the processing time $p(i)$); these in turn depend on estimations of parameters of the generating distributions (e.g. $\mathbb{E}[X \sim \mathbb{P}_j]$). The methods we use are simple. To measure how much we loose with this simplicity, we compare our methods with the ground truth on two different levels. 

First, we compare our probabilistic methods with exact \emph{clairvoyant} algorithms that rely on the knowledge of the true execution time (that the real-world scheduler clearly does not have):
the SPT (\emph{Shortest Processing Time}) strategy in the non-preemptive case and the SRPT (\emph{Shortest Remaining Processing Time}) strategy in the preemptive case. These two non-real strategies are based on full knowledge of the actual execution time of invocations that are not yet finished. This way we can analyze how strategies based on expectations approach these standard theoretically-grounded strategies for fixed processing times---and thus the limits of how much the algorithms can further gain from better estimates.

Second, our probabilistic methods estimate the parameters of distributions.
For each invocation of function $f_j$, algorithms presented in Sec.~\ref{sec:st} may require information about its expected (remaining) processing time or the expected number of invocations of function $f_j$ in the current one-minute interval.
As we want to measure the influence of the imperfection of such estimations, some of the algorithms are compared in three different variants: the \emph{reactionary} one (RE), a \emph{limited reactionary} one (RE-LIM), and the \emph{foresight} one (FOR).

In the \emph{reactionary} variant, the EPT, ERPT and $\lambda_j^k$ values for the $f_j$ function are not known and thus are estimated based on \emph{all} the previous invocations of function $f_j$. The $\lambda_j^k$ values in the reactionary model are predicted a~priori based on the actual number of invocations in the previous one-minute interval, i.e.
$$\lambda_j^{k} = \begin{cases}
    1 & \text{if $k = 1$}\\
    \# \{i \colon f(i) = j \text{ and } r(i) \in v^{k-1}\} & \text{if $k > 1$}\\
\end{cases}$$
The EPT and ERPT values are estimated as shown in Alg.~\ref{alg:SEPT}--\ref{alg:SERPT}. In particular, if it is not possible to estimate any of these values (e.g. there were no previous calls of a particular function), an arbitrary default value is used. For example, SEPT uses the average processing time of all previous invocations.

Although the storage needed to estimate the expected processing time of an invocation (see Alg.~\ref{alg:SEPT}) does not depend on the number of invocations, this is not a case when the ERPT value is calculated (see Alg.~\ref{alg:SERPT}).
In a real system, keeping information about all the previous calls of any function $f_j$ is not practical.
Therefore, we introduce \emph{limited reactionary variants} (RE-LIM) of some algorithms, in which we keep information about at most $1\,000$ last invocations of each function $f_j$.
(We also tested RE-LIM limited to $10$ and to $100$ executions which resulted in significantly worse performance).

In order to check how better estimators would impact the performance of algorithms, we compare the algorithms against \emph{foresight} (FOR) variants which use actual parameters of the distributions used to generate the instance. These parameters correspond to the perfect, clairvoyant estimations. More formally, in the foresight variants we assume that for each function $f_j$ the values of $\mathbb{E}[X \sim \mathbb{P}_j]$ (expected processing times), $\mathbb{E}[X \sim \mathbb{P}_j | X \geq p] - p$ (expected remaining processing times) and $\lambda_j^k$ are estimated a priori based on the whole instance. However, the algorithms are still probabilistic, e.g., for a just-released job we know its expected processing time, $\mathbb{E}[X \sim \mathbb{P}_j]$, but not the actual processing time $p(j)$.

\subsection{Preprocessing of the trace data}
\label{sec:ev-preprocess}
\begin{algorithm}[tb]
    \caption{Instance generation (based on Azure)}
    \begin{algorithmic}[1]
    \Function{Fill}{$T_1$, $T_2$, $m$, $\chi$, $\varepsilon$}
    \State $F \gets \{1, 2, \dots, n\}$
    \State $I \gets \{\}$ \Comment{Generated instance}
    \State $L \gets 0$ \Comment{Total load}
    \Do
    \State $j \gets \textsc{Random}(F)$
    \State $I_j \gets \text{seq. of invocation of the } f_j \text{ function}$
    \State $L_j \gets \text{total load of } I_j$
    \If{$L + L_j \leq (1 + \varepsilon)\cdot\chi\cdot m\cdot(T_2 - T_1)$}
        \State $S \gets S \cup \{j\}$
        \State $L \gets L + L_j$
        \State $I \gets I \cup I_j$
    \EndIf
    \State $F \gets F \setminus \{j\}$
    \doWhile{$L < \chi\cdot m\cdot(T_2 - T_1) \wedge |F| > 0$}
    \State \Return $I$ 
    \EndFunction
    \end{algorithmic}
    \label{alg:gen}
\end{algorithm}

We pre-processed the Azure dataset as follows.
First, we filtered out 38 functions having multiple records per day, leaving $671\,404$ out of $671\,080$ records from all the 14 days.
Then, as indicated in Sec.~\ref{sec:azure-dataset}, we removed all the records that were not related to HTTP invocations, which further narrowed the dataset to $200\,194$ records.
We were particularly interested in functions invoked by HTTP requests, as they are less regular and their invocation patterns are harder to predict and optimize --- contrary to functions triggered by internal events (e.g \texttt{cron} tasks).
Finally, we omitted functions containing missing data (i.e. missing information about execution times), which resulted in $199\,524$ records of data on $30\,325$ individual functions.

As the trace did not provide any information on either the I/O-intensiveness of individual functions, or the characteristics of I/O devices used in clusters, we assumed that functions are CPU-intensive. The given function processing times include the time needed to perform all the I/O operations. Thus, in our simulations, a processor remains busy during I/O phases. However, on the kernel level, the job performing the I/O would change its state to ``waiting'' (``not ready'') and another ready job would be assigned to the processor. As a consequence, our simulation results provide upper-bounds for what we could expect in case of I/O-intensive functions.

\subsection{Generating instances}
\label{sec:ev-method}

The performance of the scheduling algorithms was tested for various configurations.
Each configuration specified the number of available processors ($10$, $20$, $50$ or $100$), their desired average load ($70\%$, $80\%$, $90\%$ or $100\%$) and the time frame duration $T$ ($30$, $60$ or $100$ min).

For each configuration we generated $20$ independent instances (each box shows statistics over 20 instances).
For each instance, we randomly selected a window $[T_1, T_2)$ of length $T$ within one of $14$ days of the trace. This way, each instance was generated based on the data coming from a consistent interval in the trace. From within the $[T_1, T_2)$ window, we randomly selected a subset of functions so the average load $\chi$ is achieved for the given number $m$ of processors. Alg.~\ref{alg:gen} describes this process. First, we pick all functions having any invocation in the $[T_1, T_2)$ window (for clarity, we denote these functions by $1, 2, \dots, n$). Then, we randomly select functions from this set until the load of the generated instance is in the $[\chi, (1+\varepsilon)\chi]$ range (with $\varepsilon=2\%$) or the set of available functions becomes empty.
For a selected function $f_j$, we generate a sequence of its invocations using the provided information about the number of calls within each minute of the $[T_1, T_2)$ time frame and the percentiles of average execution times of its invocations during the day.
This process is fully consistent with the mapping described in Sec.~\ref{sec:model}.
The invocations of the $f_j$ function are included into the generated instance if the total load after such an inclusion does not exceed $(1+\varepsilon)\chi$ of the total available CPU time.
From this point on, we map $[T_1, T_2)$ to $[0, T)$.
The generated instance contains all the information required to evaluate the proposed algorithms -- for the $i$-th invocation we provide: the moment of call $r(i)$, a reference $f(i)$ to the invoked function and the true processing time $p(i)$. We stress that the $p(i)$ value is not revealed to the reactionary variants of the proposed algorithms, so they are required to estimate them online.

\subsection{Comparison of different algorithms}
\label{sec:ev-metrics}

\begin{figure*}[tb]
    \centering
      \subfloat[Average flow time]{{
          \includegraphics[width=0.25\textwidth,trim={0 0mm 0 0mm},clip]{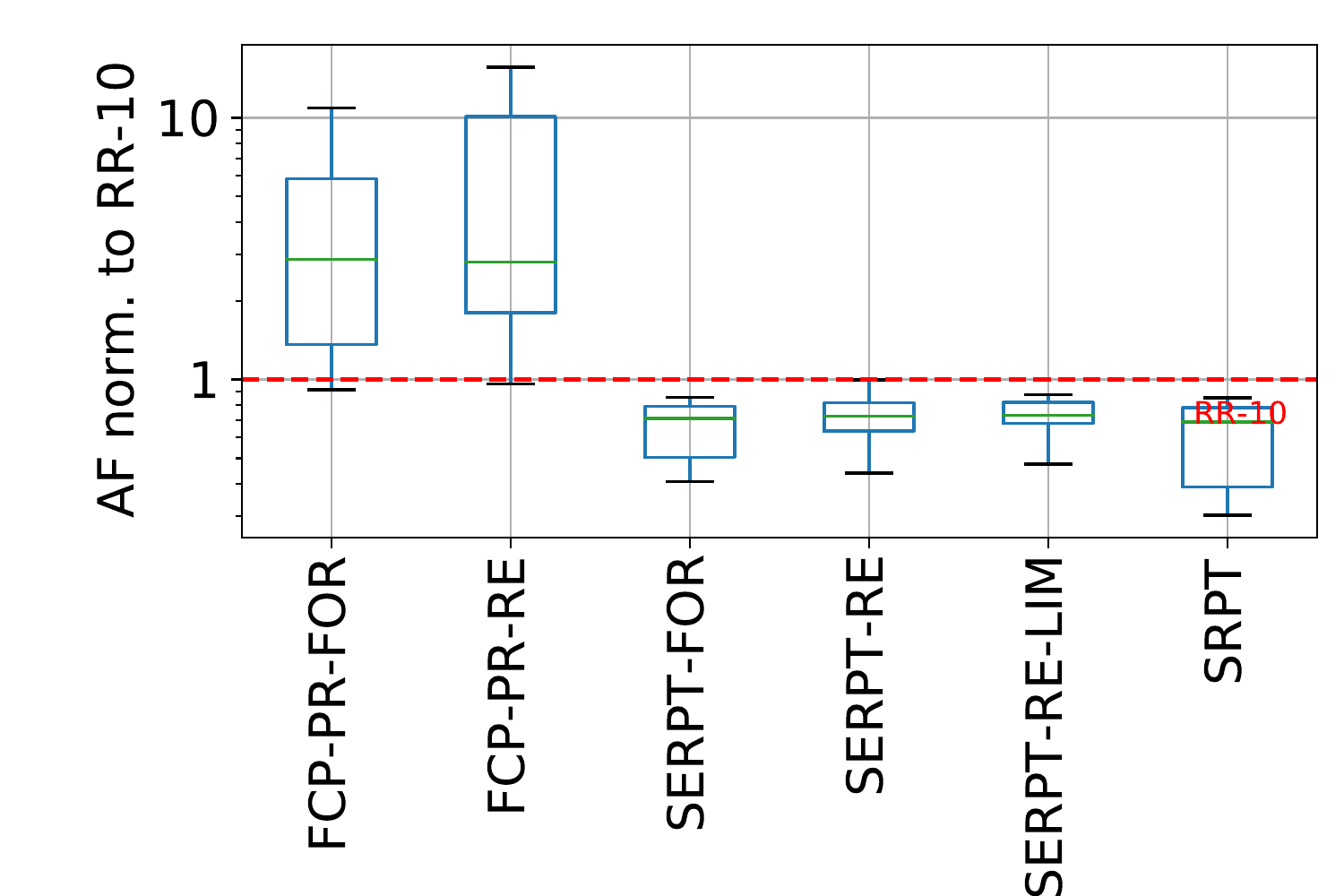}
          \includegraphics[width=0.25\textwidth,trim={0 0mm 0 0mm},clip]{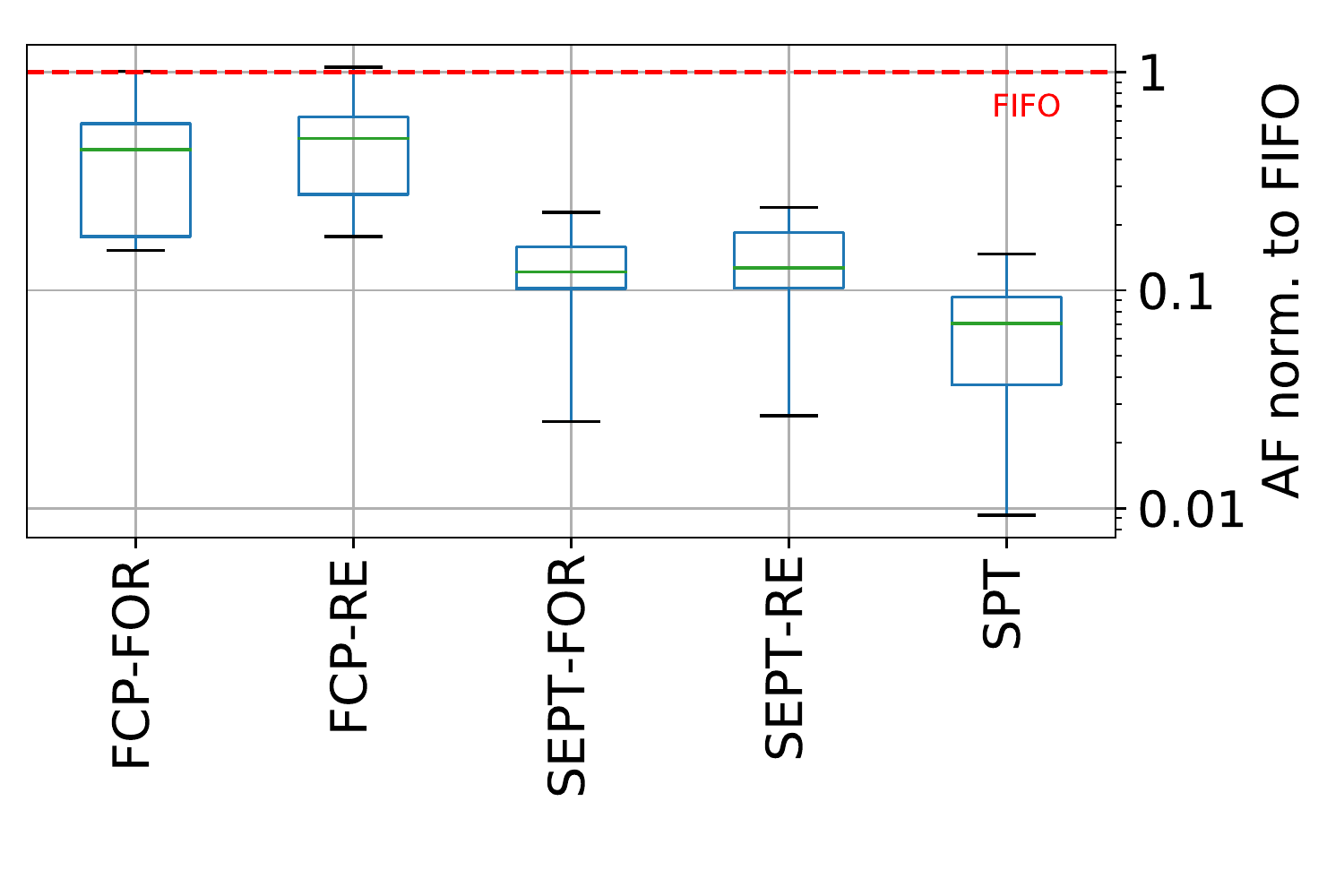}}}%
          \subfloat[Average stretch]{{
            \includegraphics[width=0.25\textwidth,trim={0 0mm 0 0mm},clip]{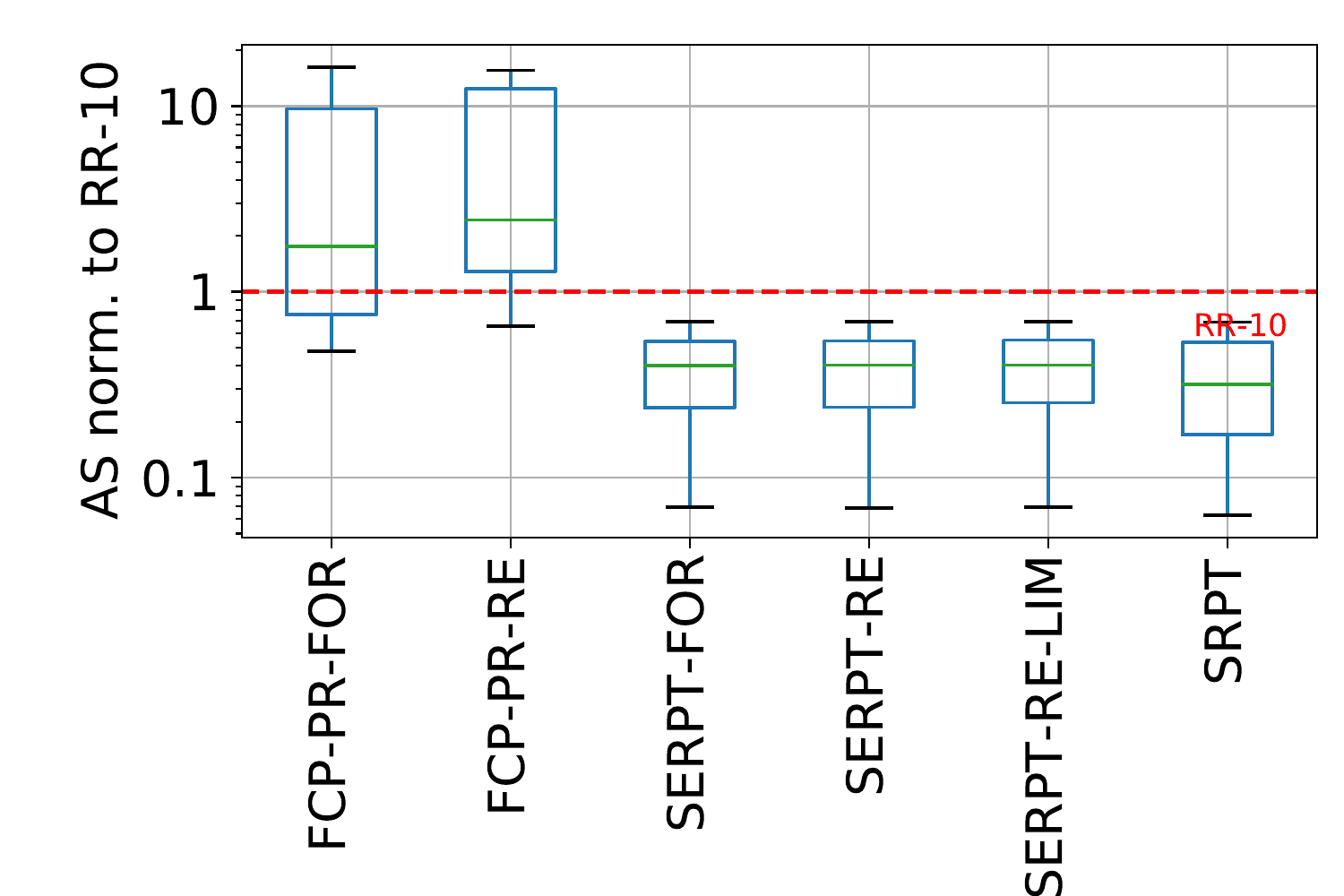}
            \includegraphics[width=0.25\textwidth,trim={0 0mm 0 0mm},clip]{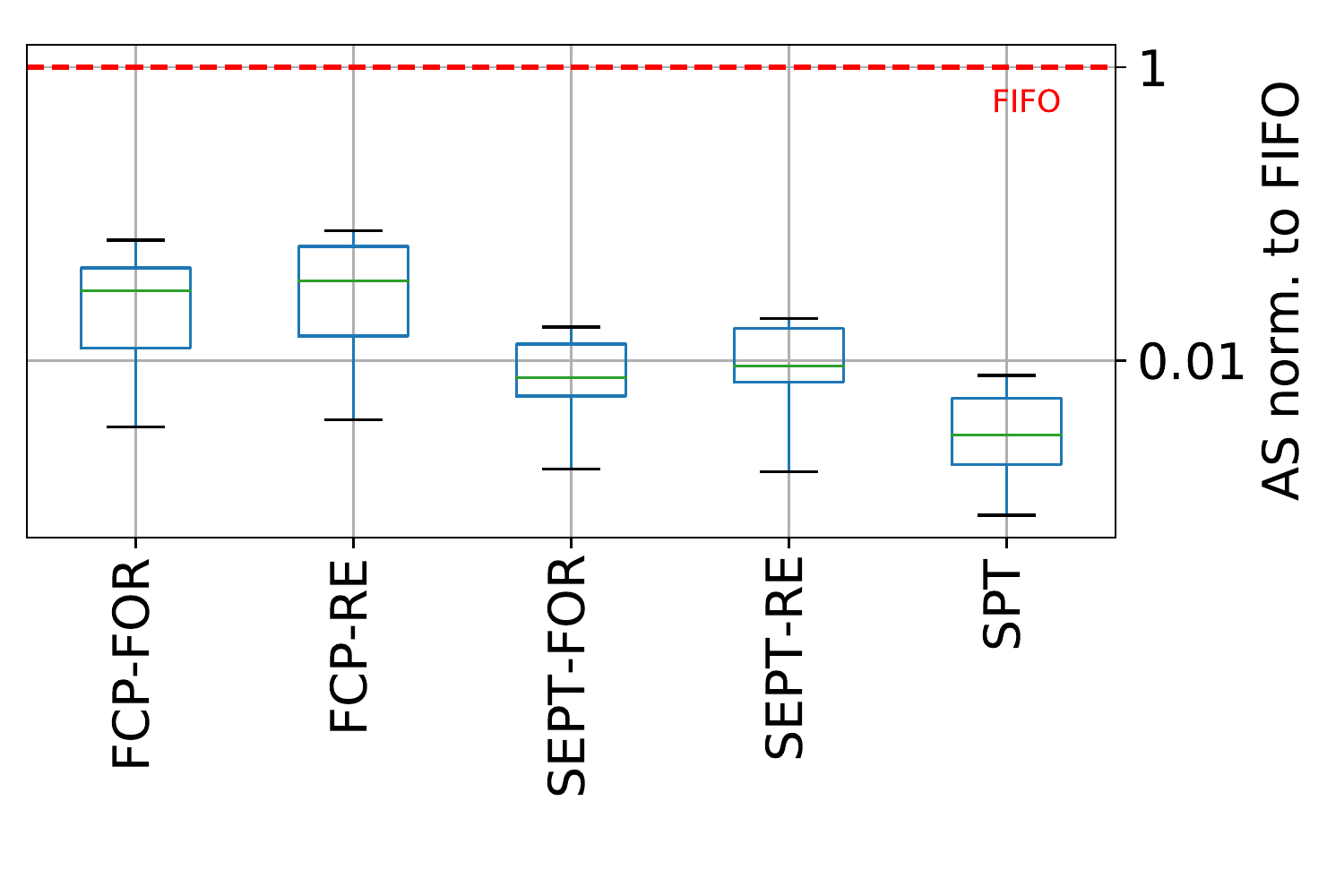}}}%
        \\
        \subfloat[99th percentile of flow time]{{
          \includegraphics[width=0.25\textwidth,trim={0 0mm 0 0mm},clip]{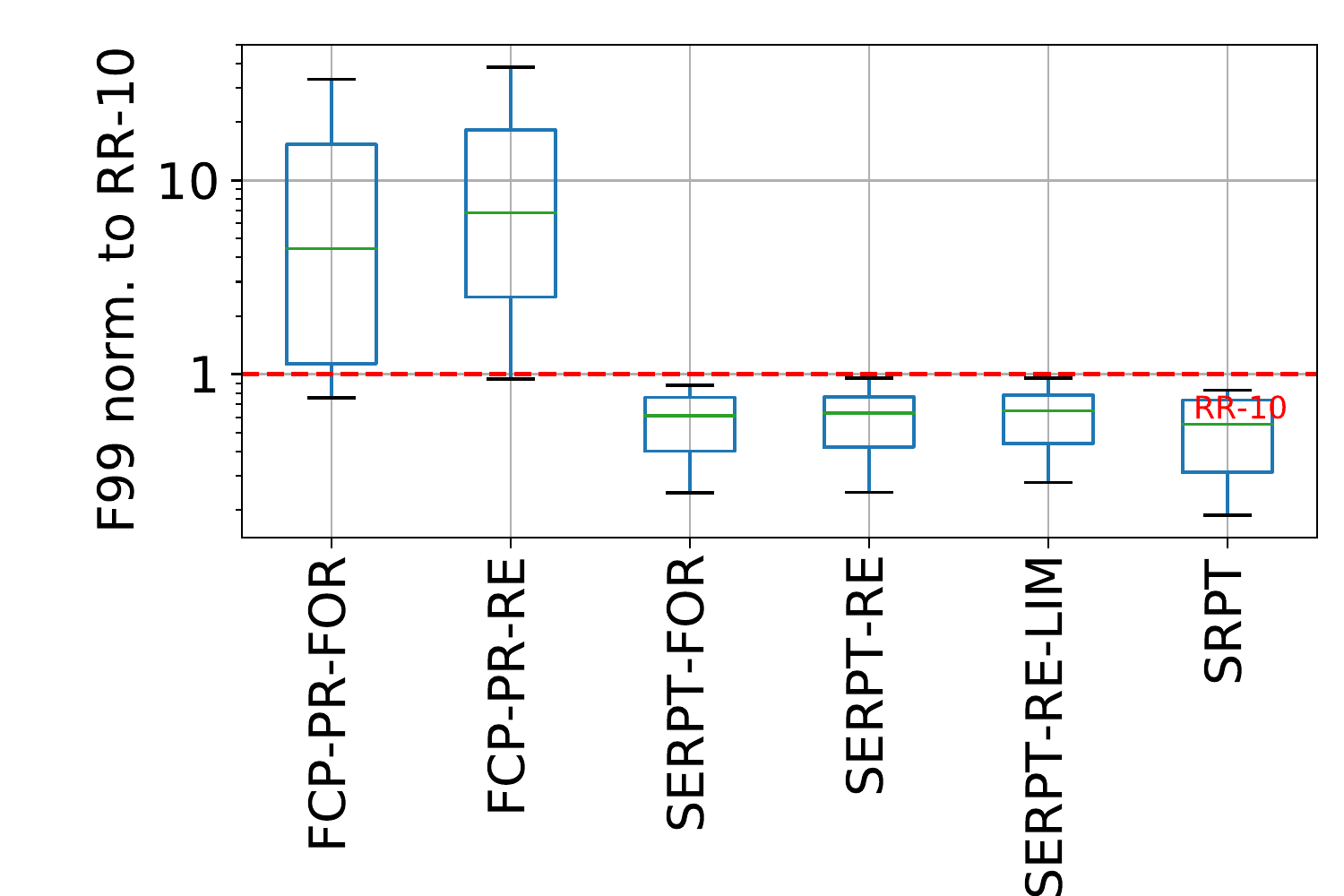}
          \includegraphics[width=0.25\textwidth,trim={0 0mm 0 0mm},clip]{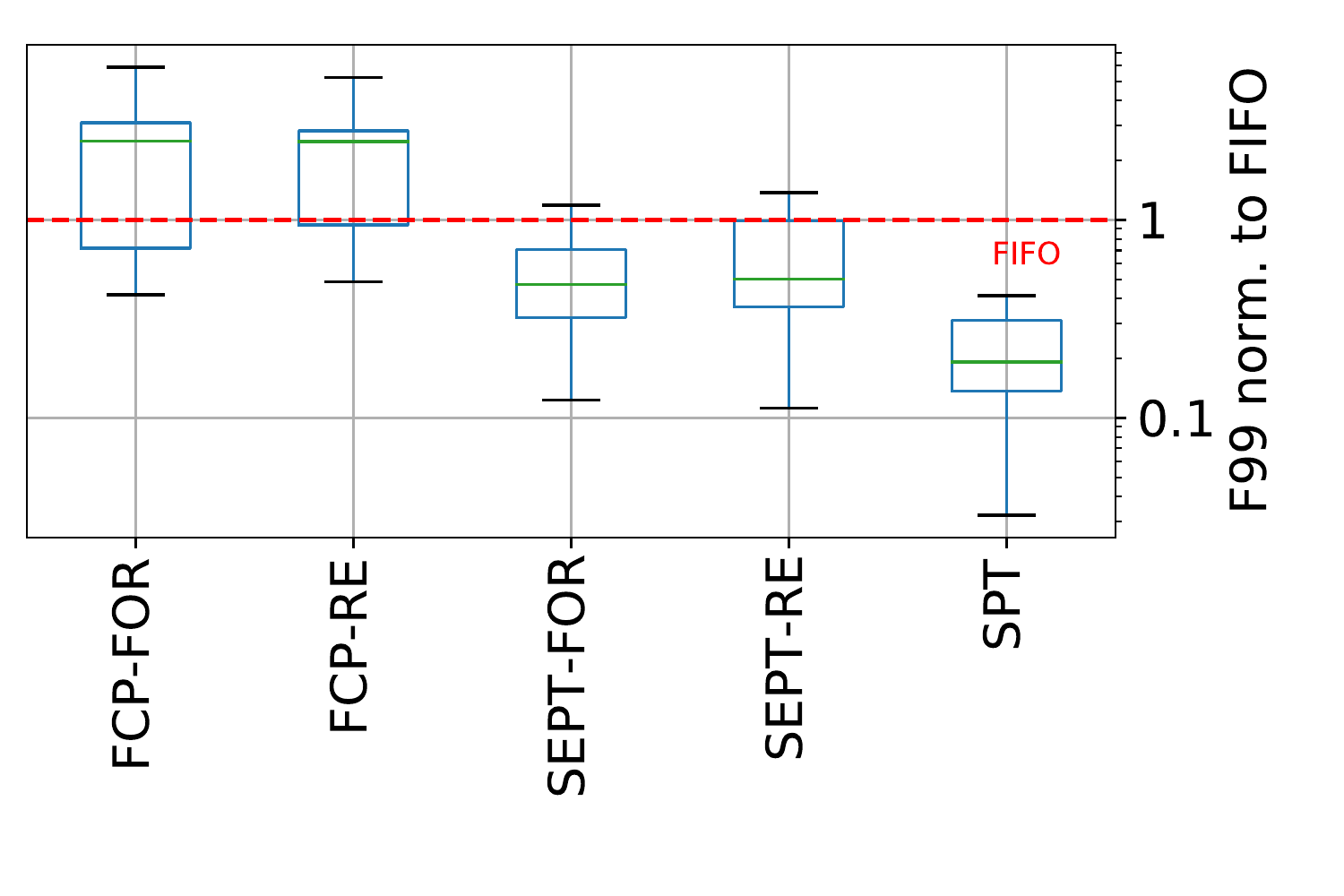}}}%
          \subfloat[99th percentile of stretch]{{
          \includegraphics[width=0.25\textwidth,trim={0 0mm 0 0mm},clip]{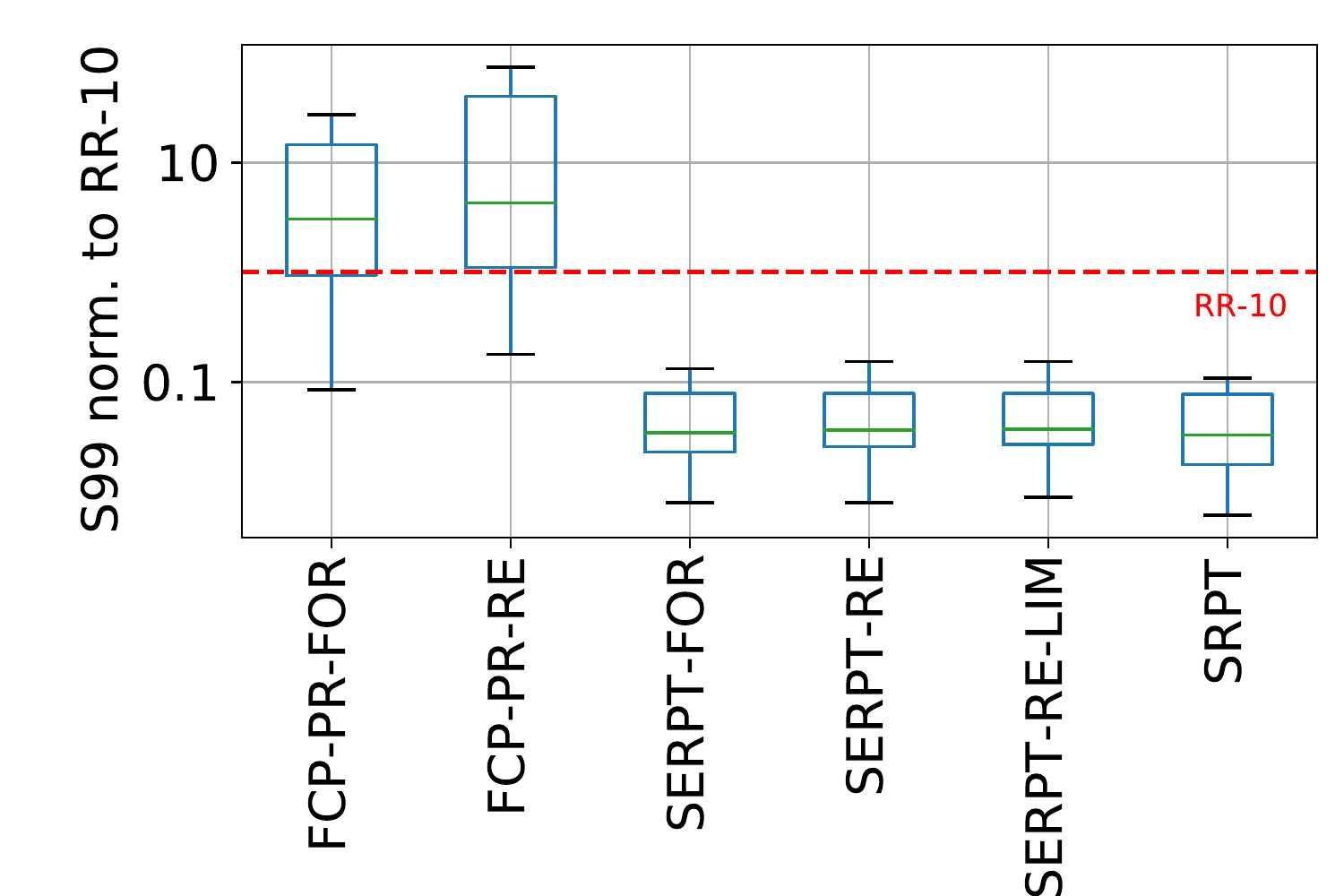}
          \includegraphics[width=0.25\textwidth,trim={0 0mm 0 0mm},clip]{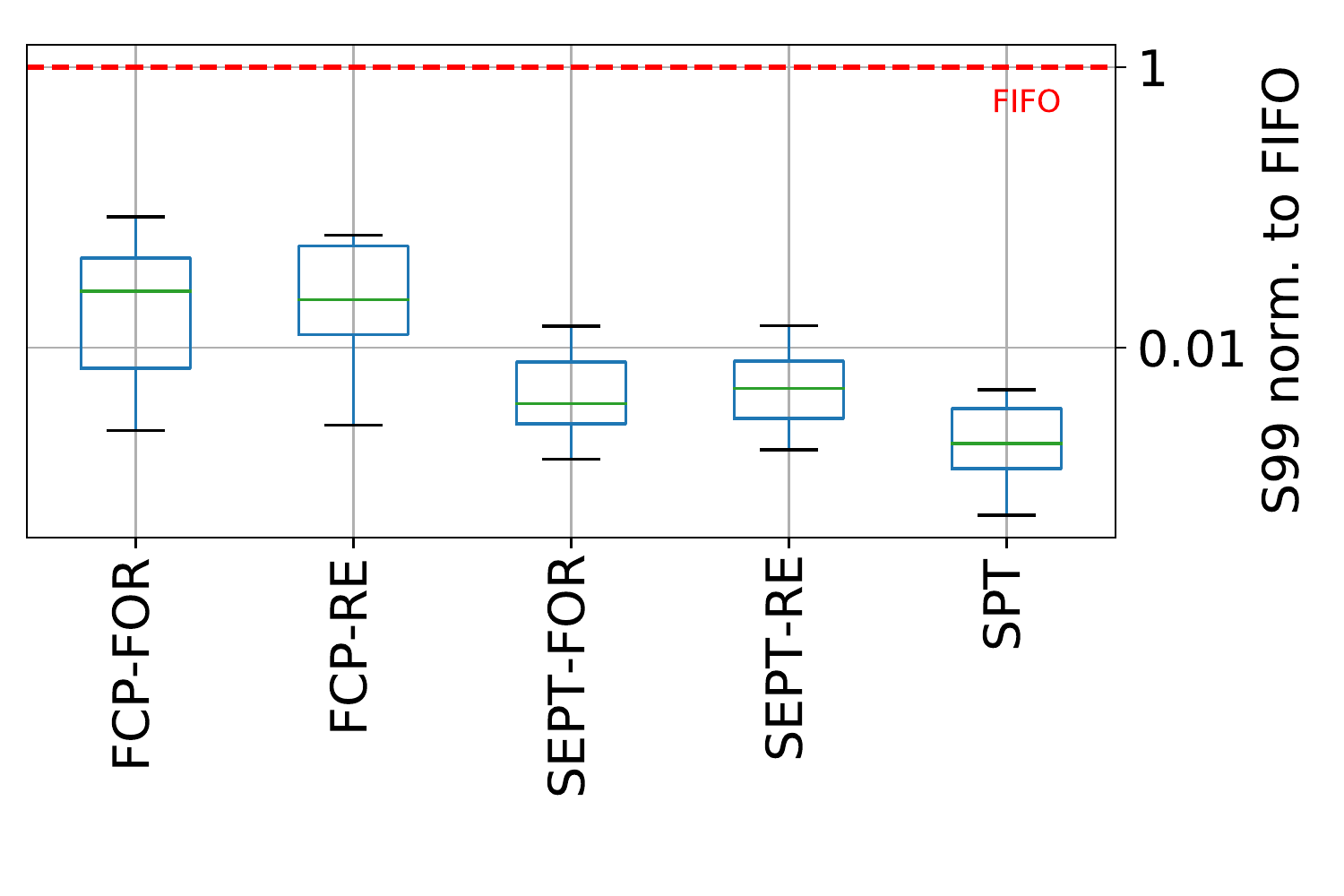}}}%
      \\
      \subfloat[Average function-aggregated flow time]{{
          \includegraphics[width=0.25\textwidth,trim={0 0mm 0 0mm},clip]{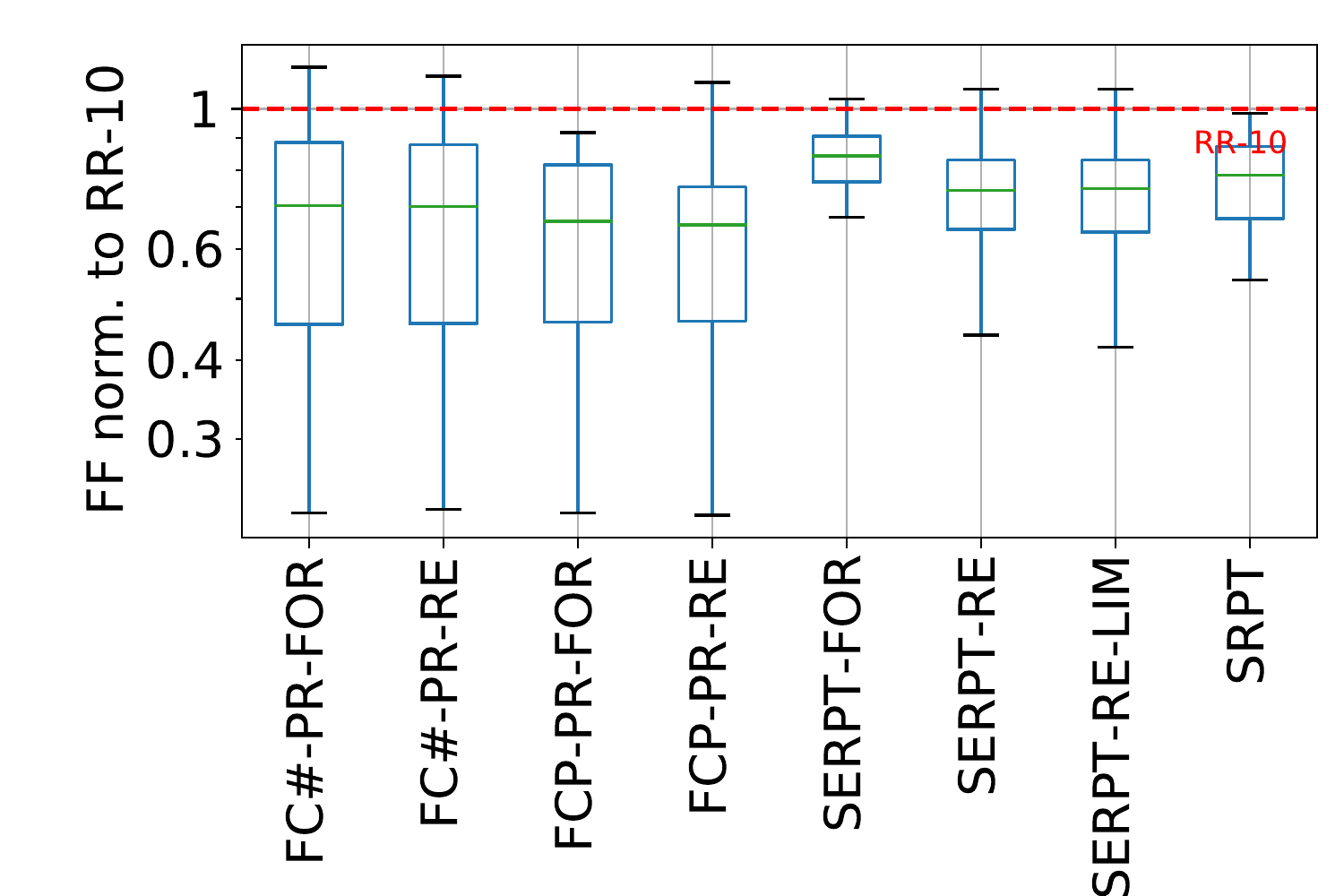}
          \includegraphics[width=0.25\textwidth,trim={0 0mm 0 0mm},clip]{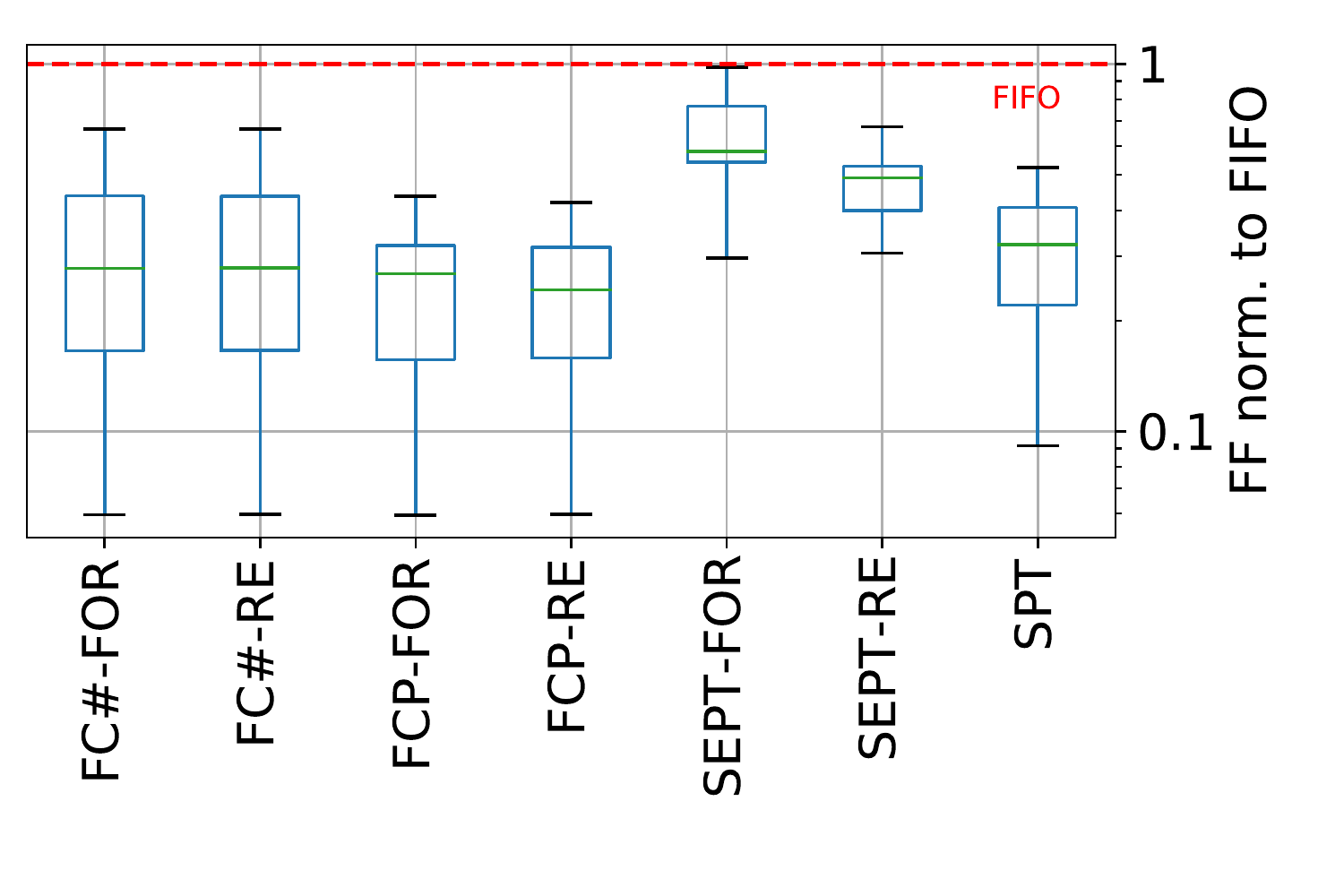}}}%
          \subfloat[Average function-aggregated stretch]{{
            \includegraphics[width=0.25\textwidth,trim={0 0mm 0 0mm},clip]{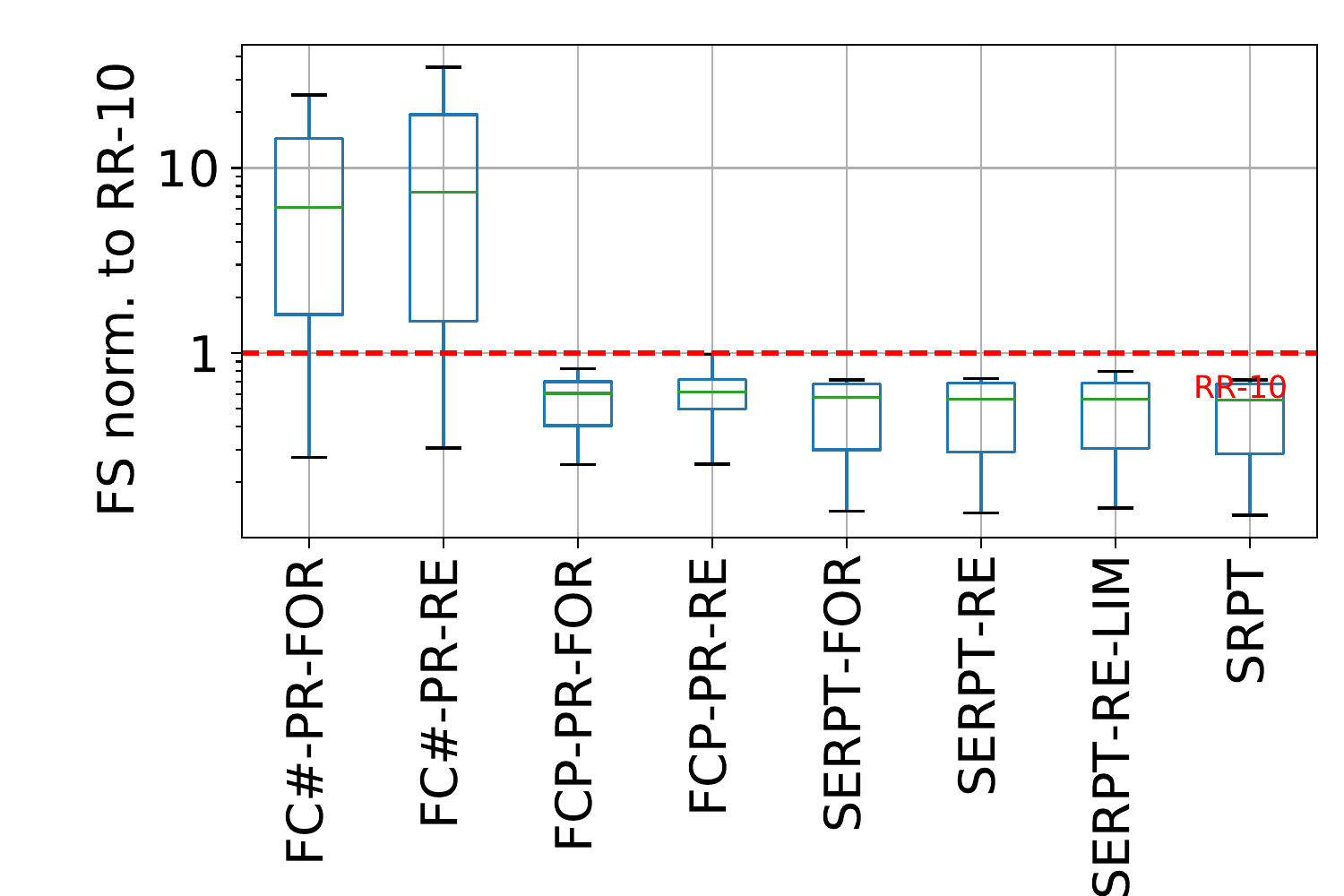}
            \includegraphics[width=0.25\textwidth,trim={0 0mm 0 0mm},clip]{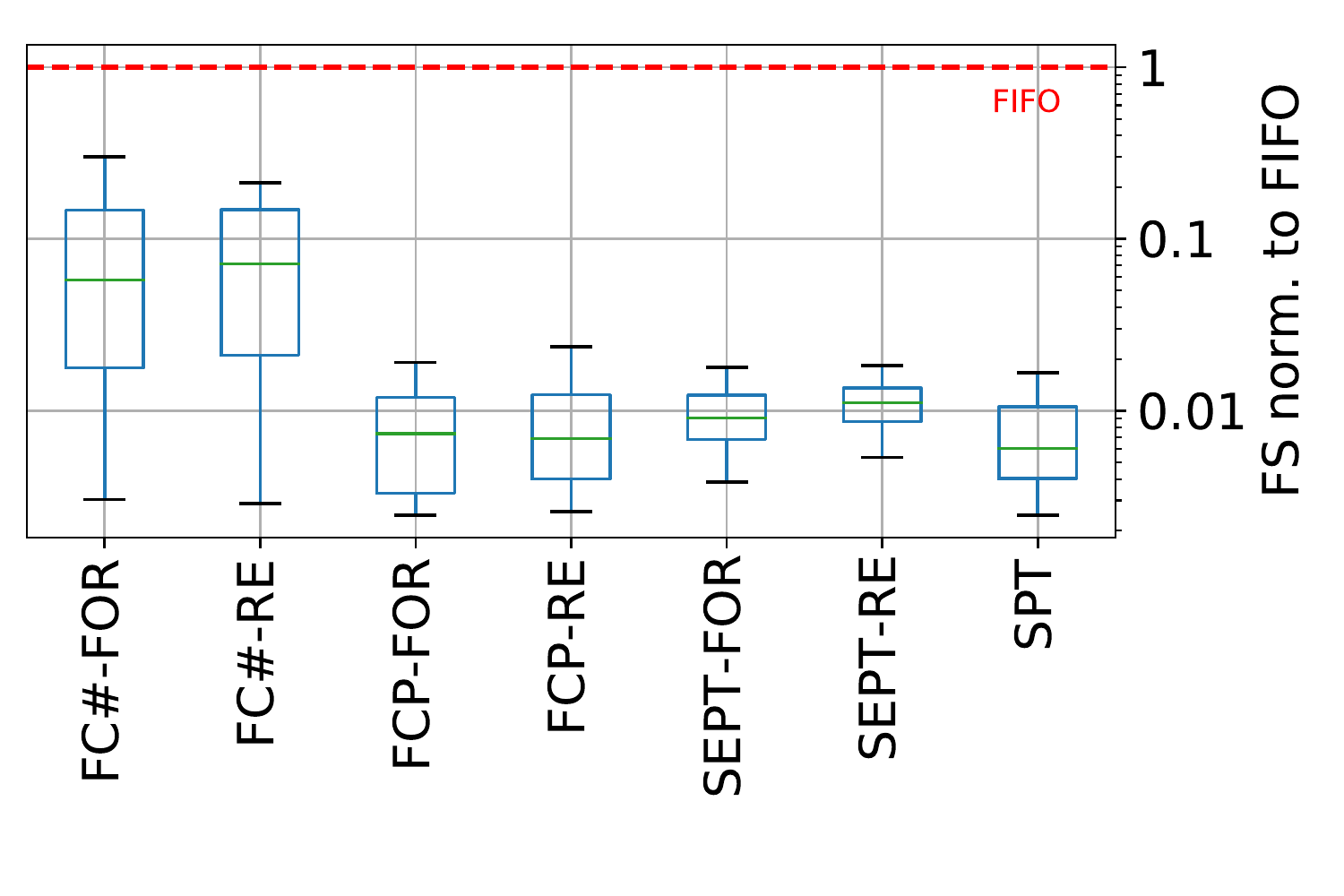}}}%
      
    \caption{Comparison of different metrics. Each box shows a statistics over 20 independent instances. Each instance has 20 processors, 30-minute time frame and 90\% load.
    }
	\label{fig:all_cmp}
\end{figure*}

Fig.~\ref{fig:all_cmp} presents the comparison of different metrics (formally defined in Sec.~\ref{sec:pd}) for configuration of $20$ processors, $90$\% average load and a $30$-minute time frame $T$.

We split results into two groups: preemptive algorithms (left side of figures) and non-preemptive algorithms (right side).
To mitigate the impact of the variability of results between instances, for each instance we \emph{normalize} the performance metric (e.g. the average flow time) by the performance of a baseline algorithm.
Results for preemptive algorithms are normalized to round-robin with a $10$-millisecond period (denoted by RR-10), i.e. metric values for each test case are divided by corresponding results for RR-10. (We also tested round-robin variants with periods of $100$ and $1\,000$ milliseconds, but they had worse results than RR-10 for all tested metrics, thus we skip them).
Similarly, results for non-preemptive algorithms are normalized to FIFO.
As RR-10 and FIFO always have normalized performance equal to 1, they are not shown on graphs.

For all considered performance metrics, our proposed SERPT and SEPT algorithms significantly improve the results compared to the baselines, round-robin and FIFO. The smallest improvements are in flow-time related metrics for preemptive case (Fig.~\ref{fig:all_cmp}, (a) and (c), left)---but, as SERPT is close to the clairvoyant SRPT, we see that there is not much space for improvement. In the non-preemptive variants (Fig.~\ref{fig:all_cmp}, (a) and (c), right), the improvements in the average flow time are almost an order of magnitude.
The reduction in stretch (Fig.~\ref{fig:all_cmp}, (b) and (d)) is larger: in non-preemptive variants by two orders of magnitude; in preemptive variants from 2-times for the average and to more than an order of magnitude for the 99th percentile. 
For all these metrics, SERPT in reactionary and foresight variants are close to SRPT, even though SRPT is clairvoyant while SERPT relies on estimates. This proves that our simple estimates of processing times are sufficient.
However, in non-preemptive cases, the difference between SEPT and the clairvoyant SPT is larger: here, the impact of a wrong processing time estimate is harder to correct.
SERPT limited to 1000 last executions (SERPT-RE-LIM) performed similarly to SERPT-RE for all tested metrics, which is promising, as that variant requires less memory when implemented in a real-world scheduler.
For our fair, function-aggregated metrics (Fig.~\ref{fig:all_cmp}, (e) and (f)), FCP dominates other variants including FC\# (which we skip from other figures as it was always dominated by FCP)---with an exception of average stretch in the clairvoyant variant, where FCP performance is similar to SERPT.

\subsection{Impact of instance parameters}
\label{sec:ev-params}

\begin{figure*}[tb!]
    \centering
      \subfloat[FIFO]{{\includegraphics[width=0.16\textwidth,trim={0 2mm 0 6mm},clip]{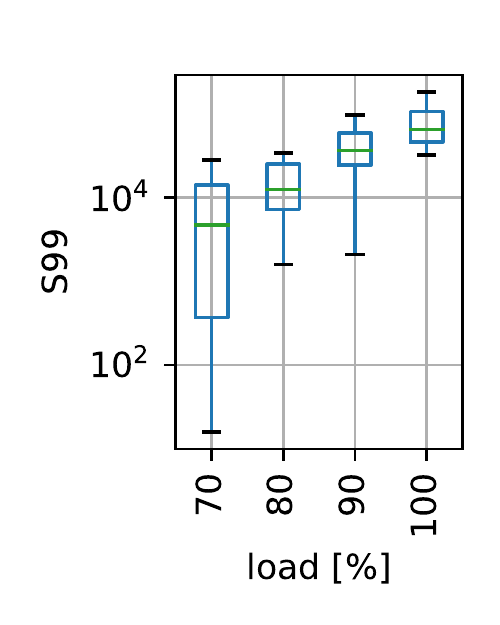}}}%
      \subfloat[RR-10]{{\includegraphics[width=0.16\textwidth,trim={0 2mm 0 6mm},clip]{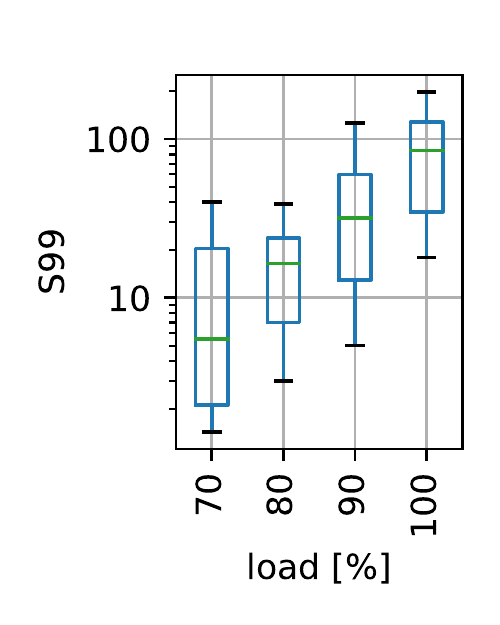}}}%
      \subfloat[SERPT-RE-LIM]{{\includegraphics[width=0.16\textwidth,trim={0 2mm 0 6mm},clip]{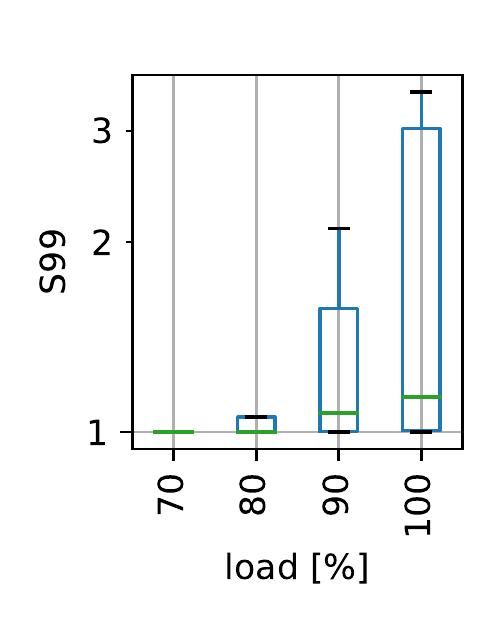}}}%
      \subfloat[SERPT-RE]{{\includegraphics[width=0.16\textwidth,trim={0 2mm 0 6mm},clip]{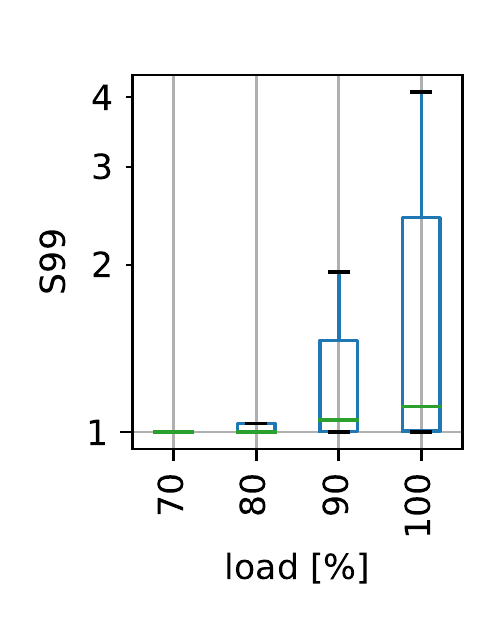}}}%
      \subfloat[SRPT]{{\includegraphics[width=0.16\textwidth,trim={0 2mm 0 6mm},clip]{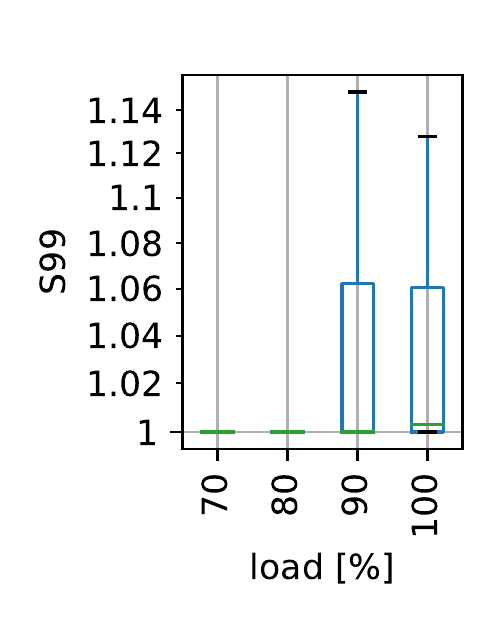}}}%
      \subfloat[SPT]{{\includegraphics[width=0.16\textwidth,trim={0 2mm 0 6mm},clip]{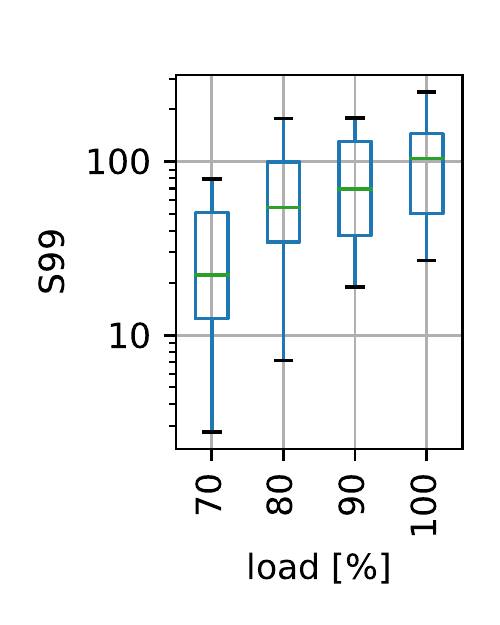}}}%
      \caption{99th percentile of stretch when varying the average load. 30-minute time frame, 20 processors.}
	\label{fig:cmp_load}
\end{figure*}

\begin{figure*}[tb!]
    \centering
      \subfloat[SEPT-RE, average flow]{{\includegraphics[width=0.25\textwidth,trim={0 2mm 0 5mm},clip]{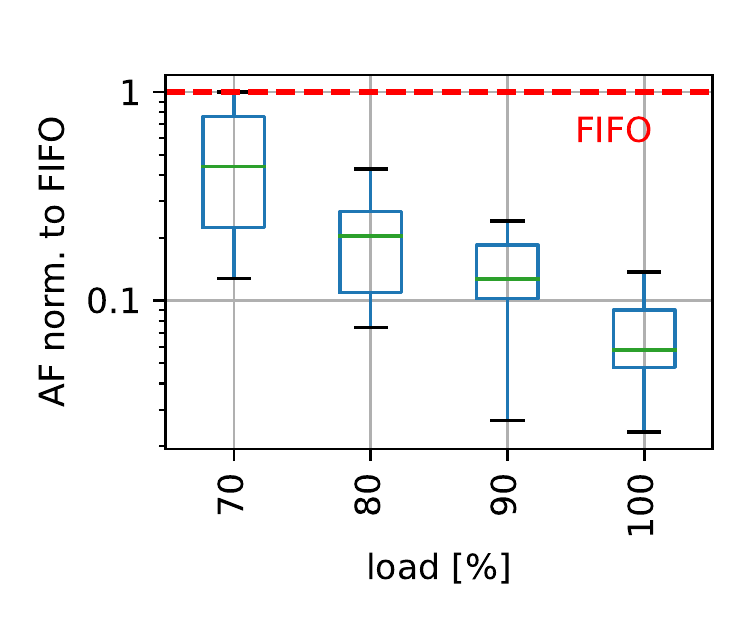}}}%
      \subfloat[SERPT-RE-LIM, average flow]{{\includegraphics[width=0.25\textwidth,trim={0 2mm 0 5mm},clip]{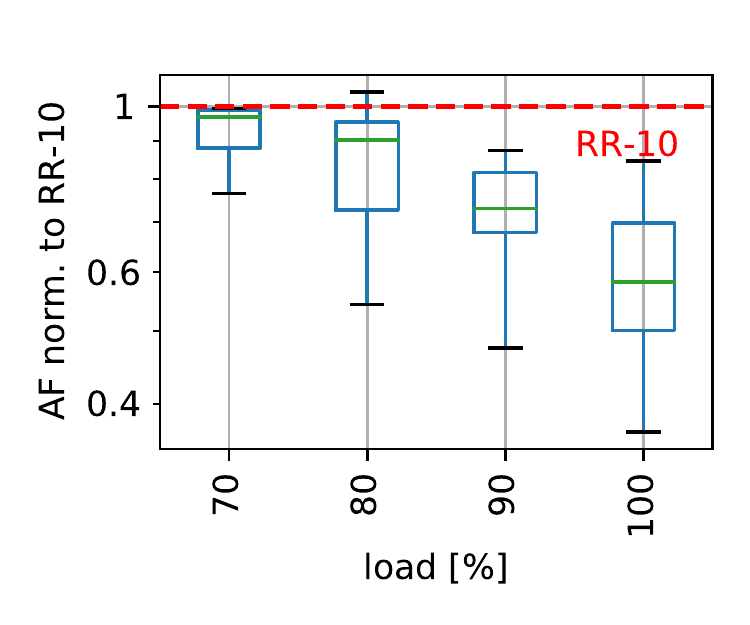}}}%
      \subfloat[SEPT-RE, average stretch]{{\includegraphics[width=0.25\textwidth,trim={0 2mm 0 5mm},clip]{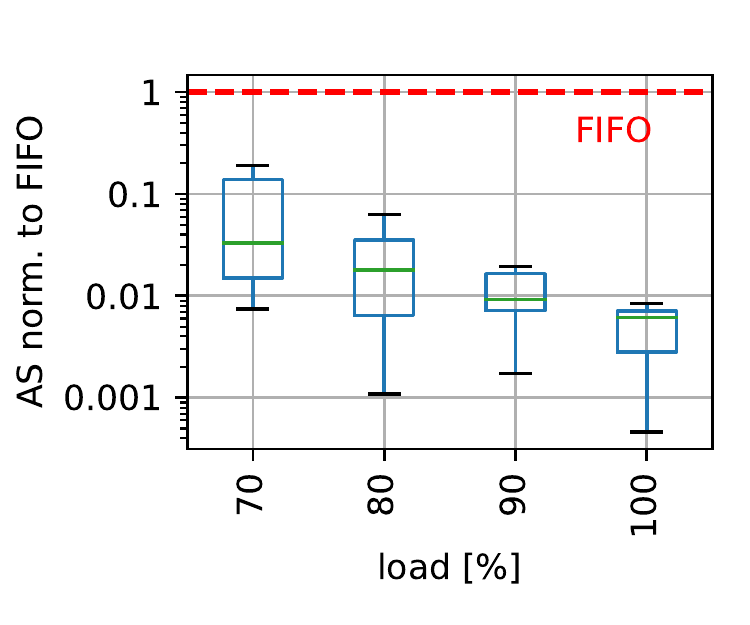}}}%
      \subfloat[SERPT-RE-LIM, average stretch]{{\includegraphics[width=0.25\textwidth,trim={0 2mm 0 5mm},clip]{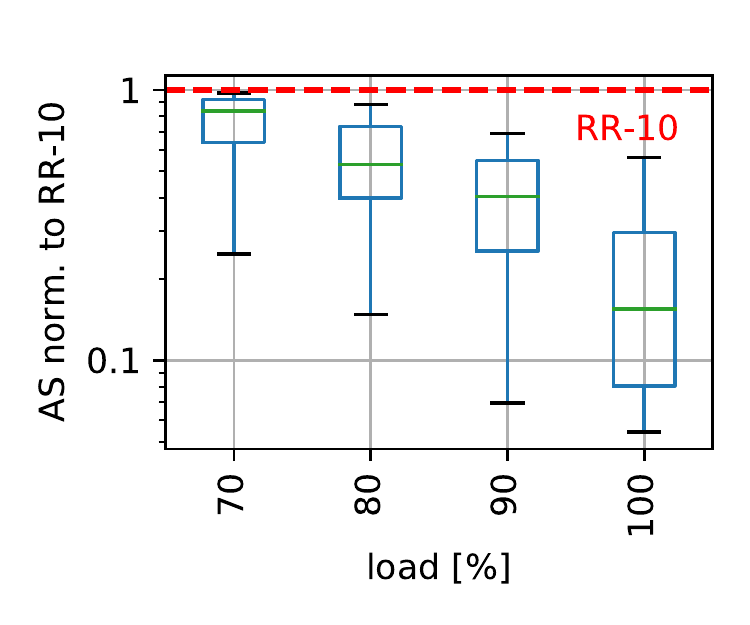}}}%
      \\
      \subfloat[FC\#-RE, FS]{{\includegraphics[width=0.25\textwidth,trim={0 2mm 0 5mm},clip]{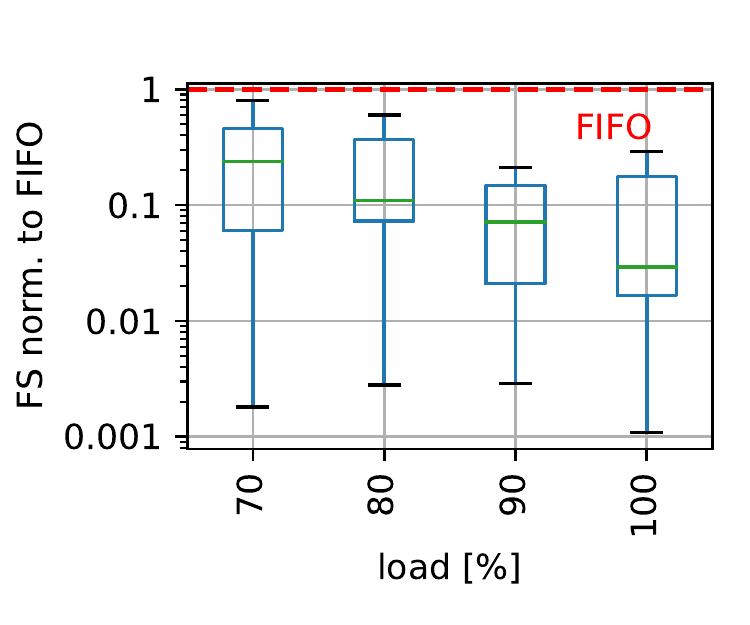}}}%
      \subfloat[FC\#-RE, FF]{{\includegraphics[width=0.25\textwidth,trim={0 2mm 0 5mm},clip]{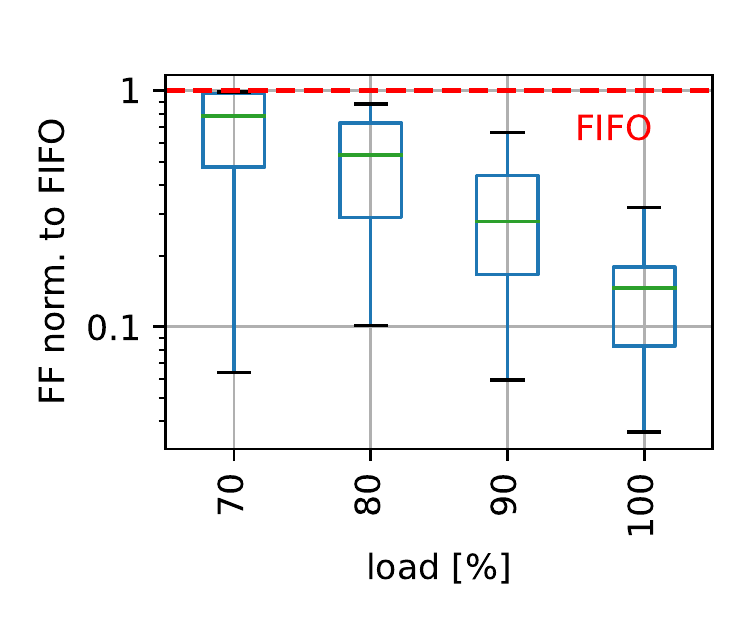}}}%
      \subfloat[FCP-RE, FS]{{\includegraphics[width=0.25\textwidth,trim={0 2mm 0 5mm},clip]{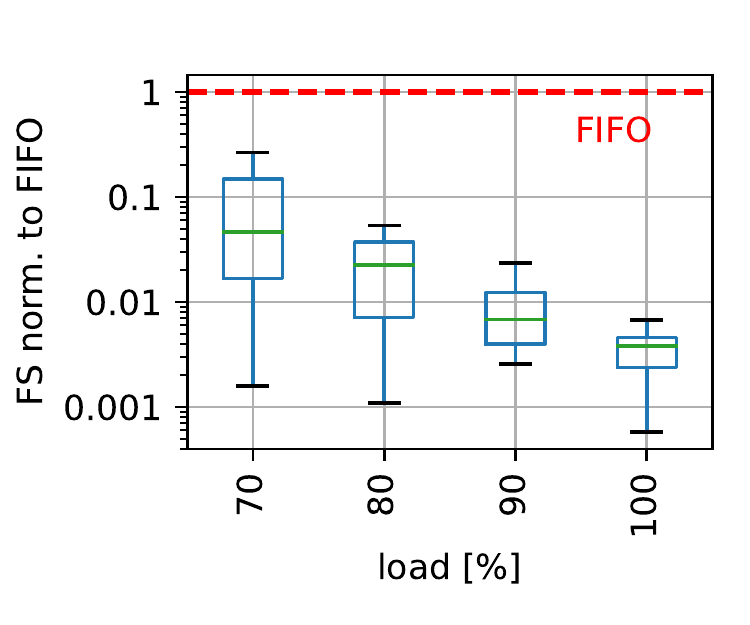}}}%
      \subfloat[FCP-RE, FF]{{\includegraphics[width=0.25\textwidth,trim={0 2mm 0 5mm},clip]{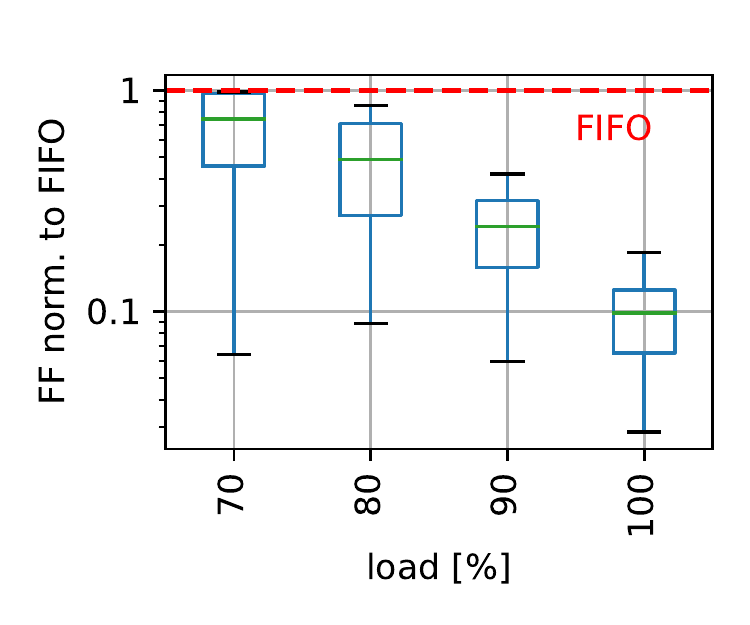}}}%
      \caption{Relative performance when varying the average load. 30 min time frame, 20 processors.}
	\label{fig:cmp_load_fc}
\end{figure*}

To make sure that the obtained results are valid for a wide range of scenarios, we verified the impact of changing average loads, processor counts and time window lengths.

Fig.~\ref{fig:cmp_load} presents the 99th percentile of stretch (S99) with different loads. We chose this metric as it is the most sensitive to the density of function calls.
First, for all loads SRPT results are close to the optimal 1, demonstrating that in all cases it is feasible to pack invocations almost optimally. Second, stretch increases with load for all other algorithms---however, both the increase and the absolute numbers are larger for the baselines FIFO and RR, compared with SERPT. Fig.~\ref{fig:cmp_load_fc} reinforces this observation: the higher the load, the better is the performance of our algorithms compared to the baselines.

\begin{figure*}[tb!]
    \centering
    \subfloat[SEPT-RE, average flow]{{\includegraphics[width=0.25\textwidth,trim={0 2mm 0 5mm},clip]{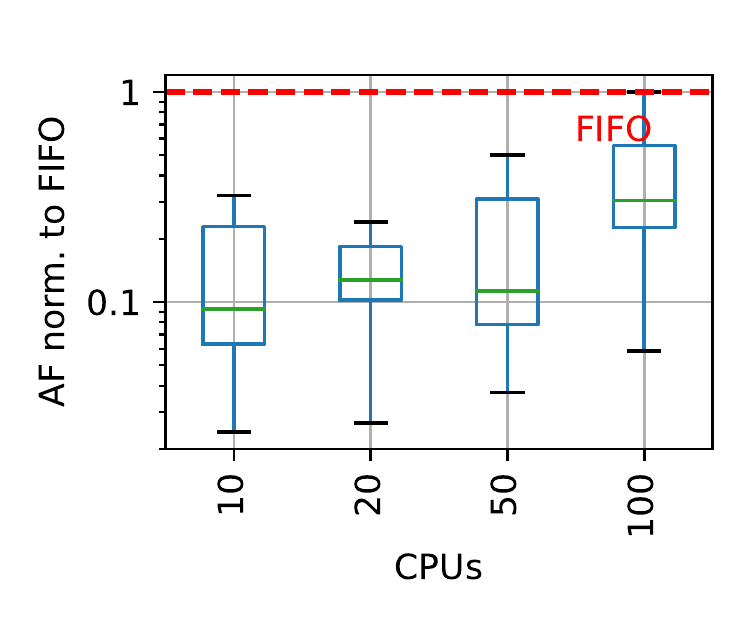}}}%
    \subfloat[SERPT-RE-LIM, average flow]{{\includegraphics[width=0.25\textwidth,trim={0 2mm 0 5mm},clip]{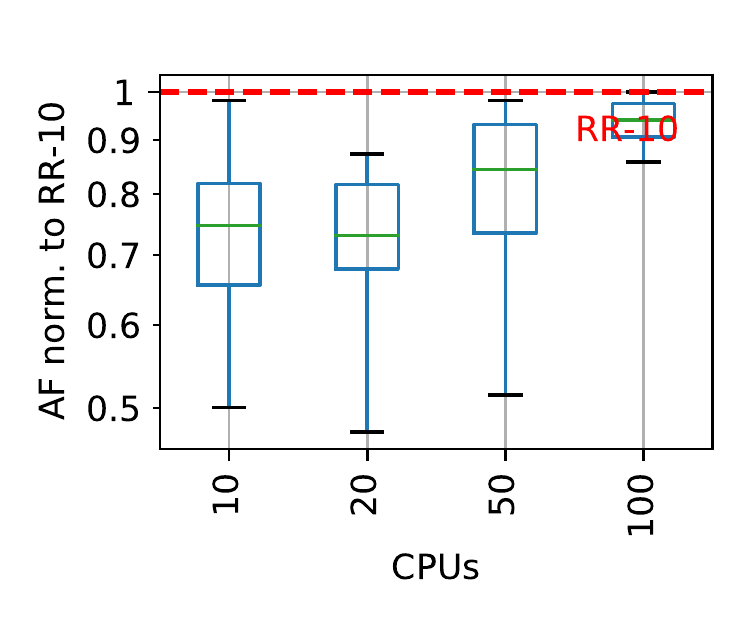}}}%
    \subfloat[SEPT-RE, average stretch]{{\includegraphics[width=0.25\textwidth,trim={0 2mm 0 5mm},clip]{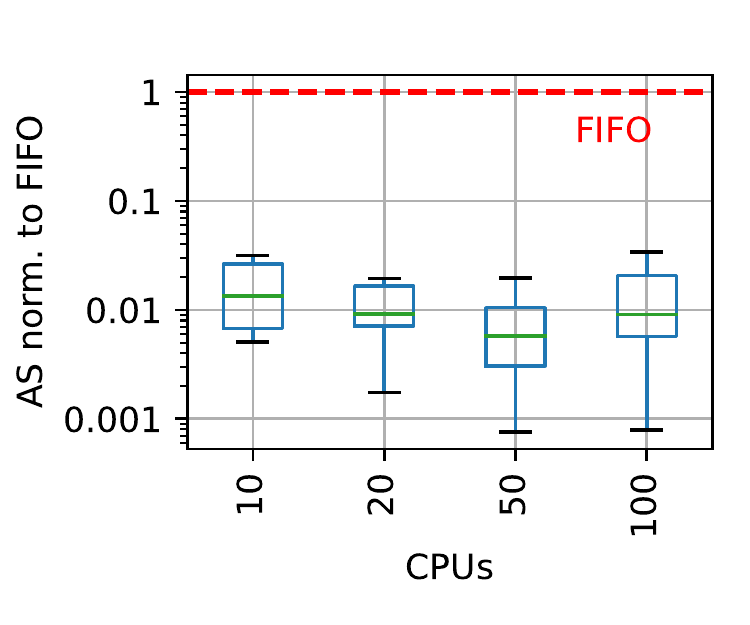}}}%
    \subfloat[SERPT-RE-LIM, average stretch]{{\includegraphics[width=0.25\textwidth,trim={0 2mm 0 5mm},clip]{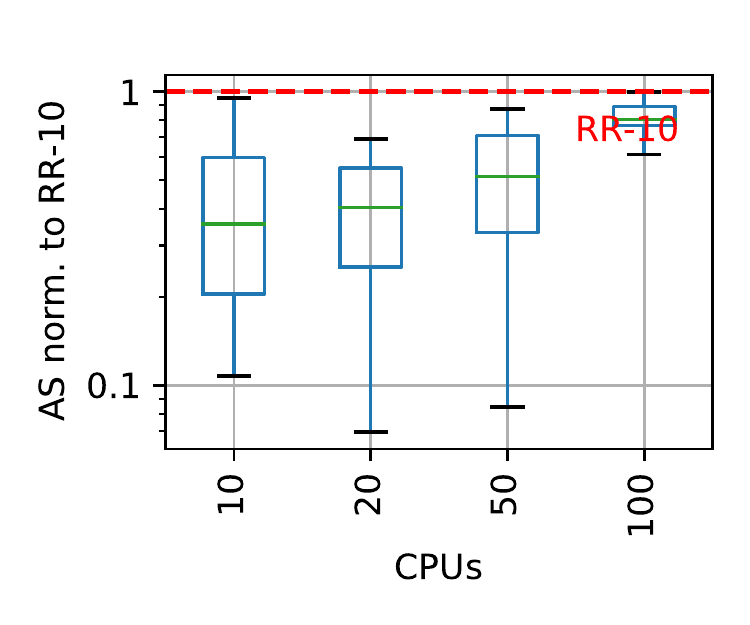}}}%
    \\
    \subfloat[SPT, average stretch]{{\includegraphics[width=0.25\textwidth,trim={0 2mm 0 5mm},clip]{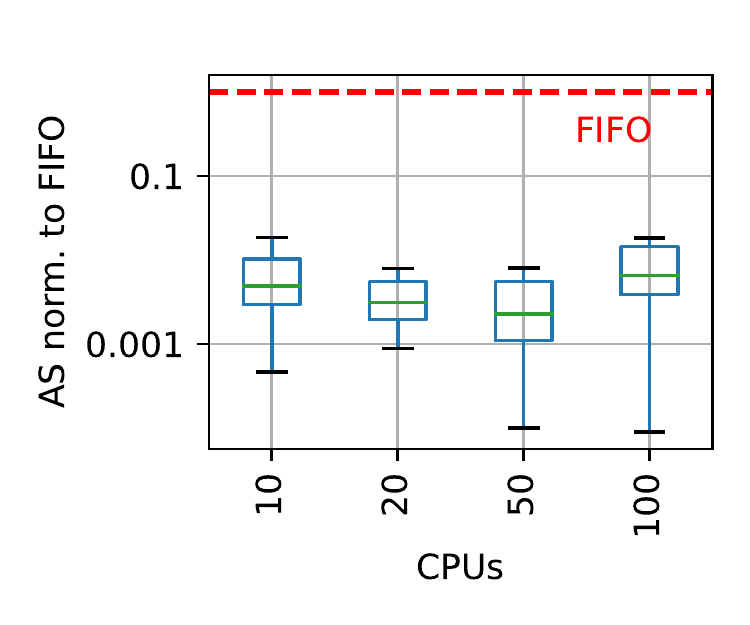}}}%
    \subfloat[SRPT, average stretch]{{\includegraphics[width=0.25\textwidth,trim={0 2mm 0 5mm},clip]{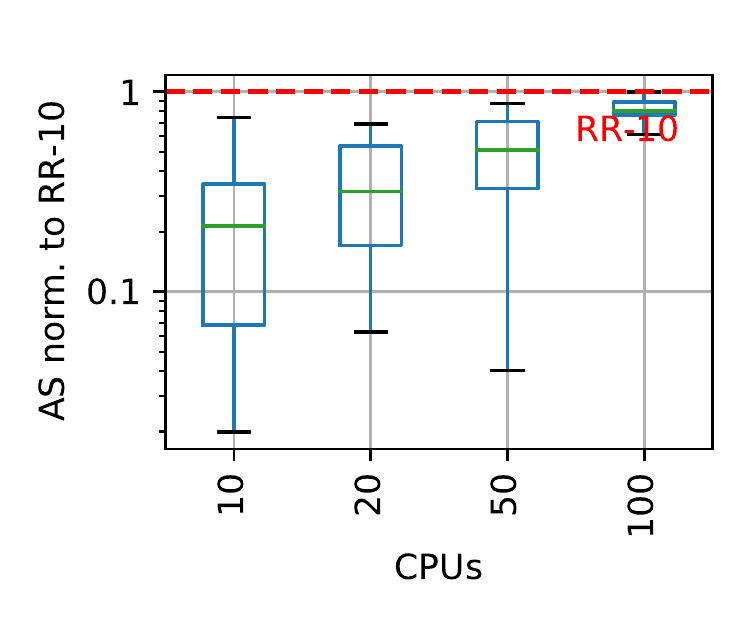}}}%
    \subfloat[FCP-RE + FS]{{\includegraphics[width=0.25\textwidth,trim={0 2mm 0 5mm},clip]{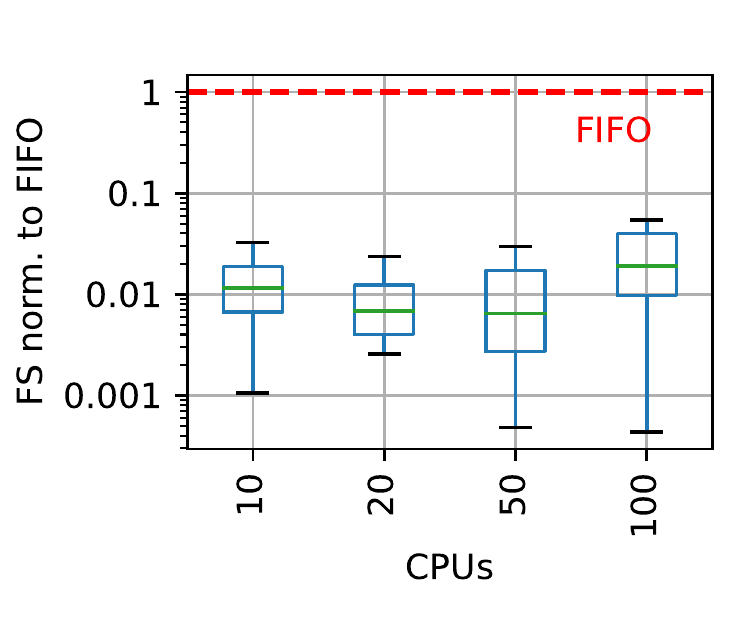}}}%
    \subfloat[FCP-RE + FF]{{\includegraphics[width=0.25\textwidth,trim={0 2mm 0 5mm},clip]{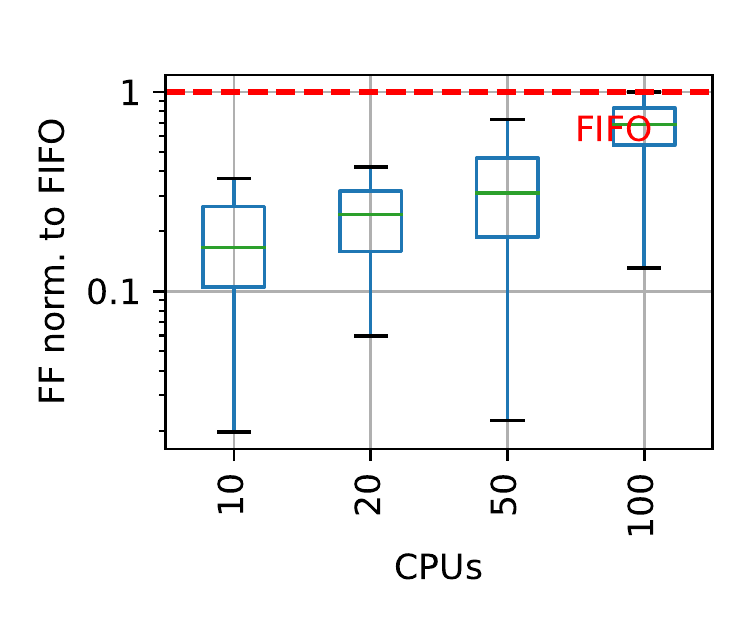}}}%
    \caption{Relative performance when varying the number of processors. 30-minute time frame and 90\% load.}
	\label{fig:cmp_cpu_count}
\end{figure*}

We also analyzed how results change when the number of processors changes, with up to 100 processors (as the largest C2 instance in AWS has 96).
Fig.~\ref{fig:cmp_cpu_count} shows that with the increase in the number of processors, it is easier to schedule invocations almost optimally even with simple heuristics, as it is less and less probable that all processors will be blocked on processing long invocations---thus, the impact of better scheduling methods diminishes. This is confirmed by smaller relative gains of the SRPT in the preemptive case (Fig.~\ref{fig:cmp_cpu_count}, (f)).

\begin{figure*}
    \centering
    \subfloat[FC\#-RE + FS]{{\includegraphics[width=0.25\textwidth,trim={0 2mm 0 5mm},clip]{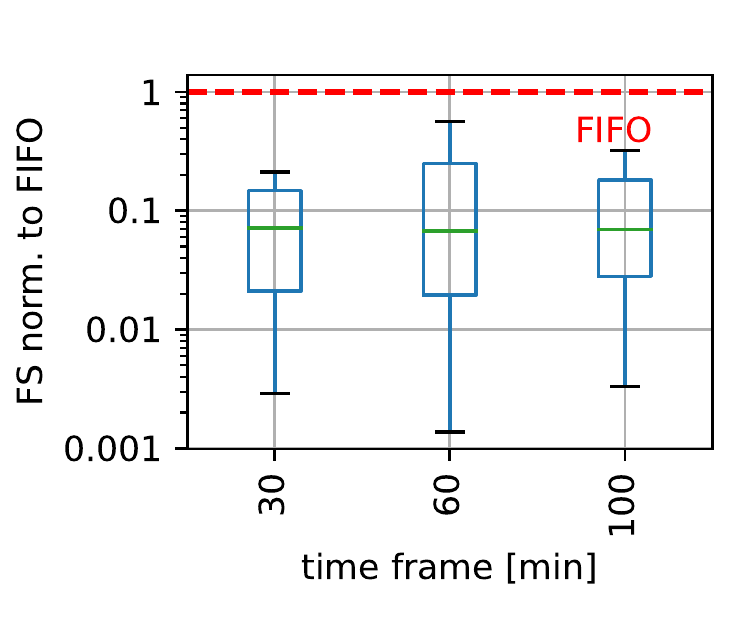}}}%
    \subfloat[FC\#-RE + FF]{{\includegraphics[width=0.25\textwidth,trim={0 2mm 0 5mm},clip]{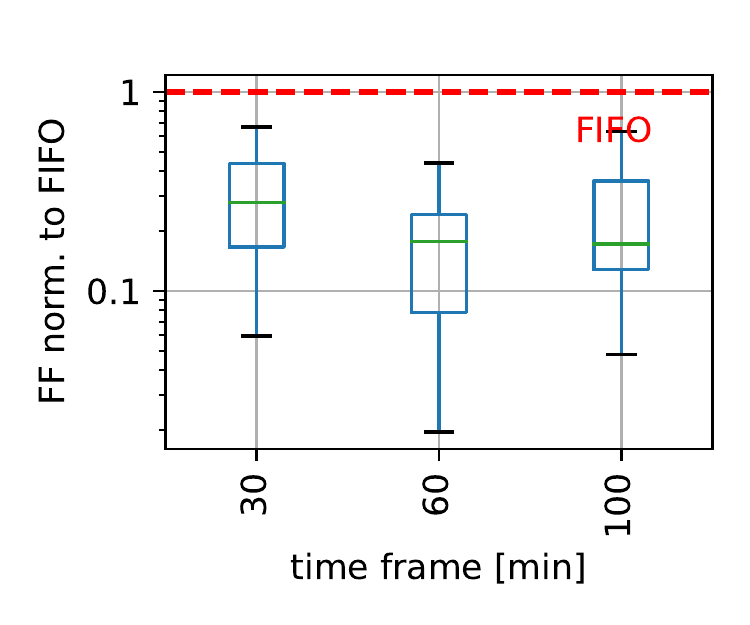}}}%
    \subfloat[SEPT-RE + AS]{{\includegraphics[width=0.25\textwidth,trim={0 2mm 0 5mm},clip]{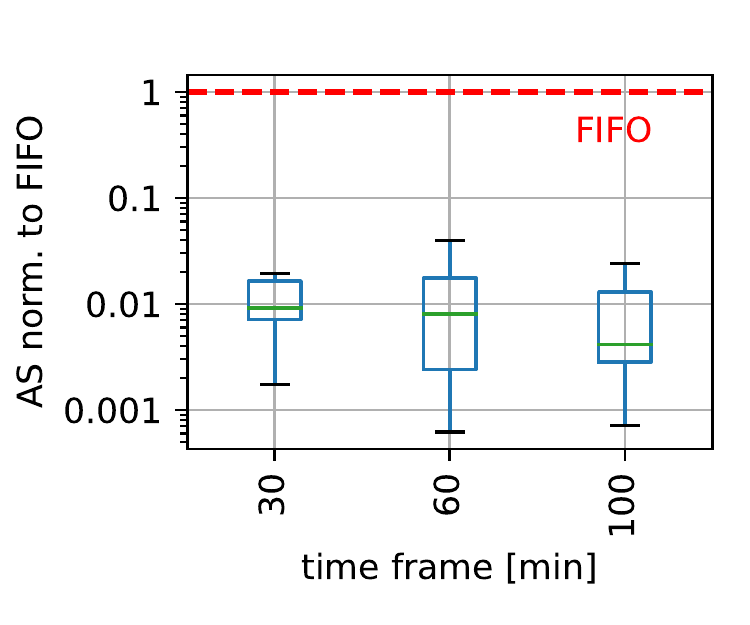}}}%
    \subfloat[SERPT-RE-LIM + AF]{{\includegraphics[width=0.25\textwidth,trim={0 2mm 0 5mm},clip]{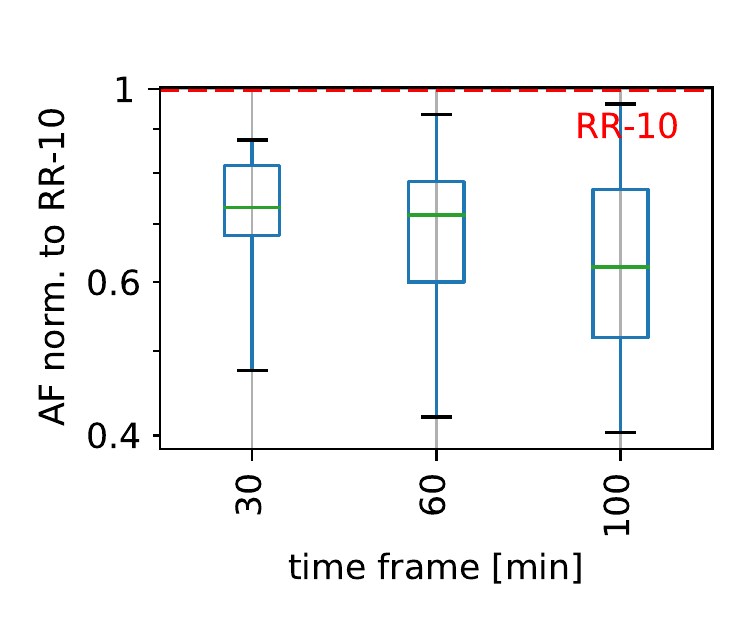}}}%
    \\
    \subfloat[FCP-RE + FS]{{\includegraphics[width=0.25\textwidth,trim={0 2mm 0 5mm},clip]{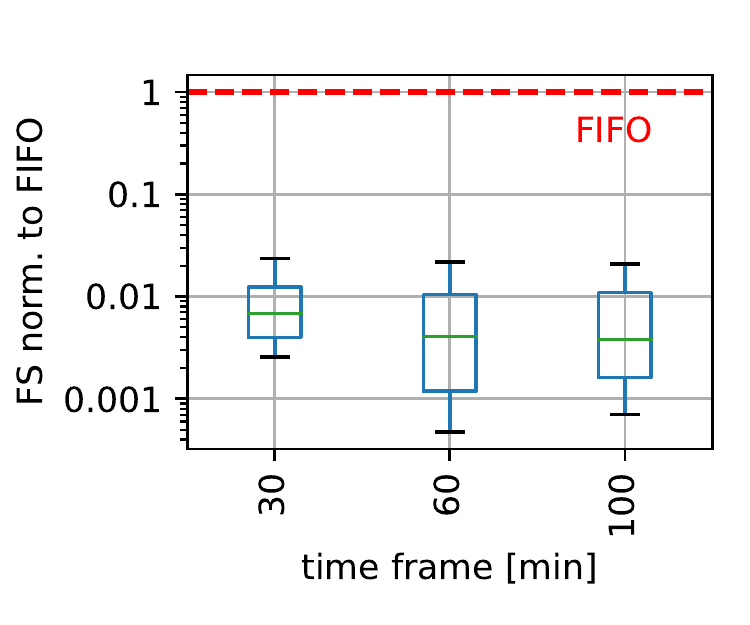}}}%
    \subfloat[FCP-RE + FF]{{\includegraphics[width=0.25\textwidth,trim={0 2mm 0 5mm},clip]{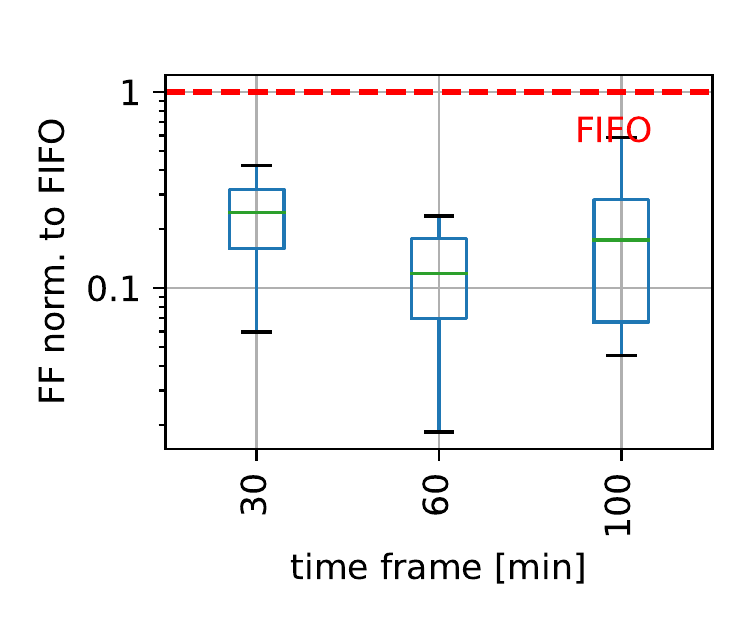}}}%
    \subfloat[SEPT-RE + AF]{{\includegraphics[width=0.25\textwidth,trim={0 2mm 0 5mm},clip]{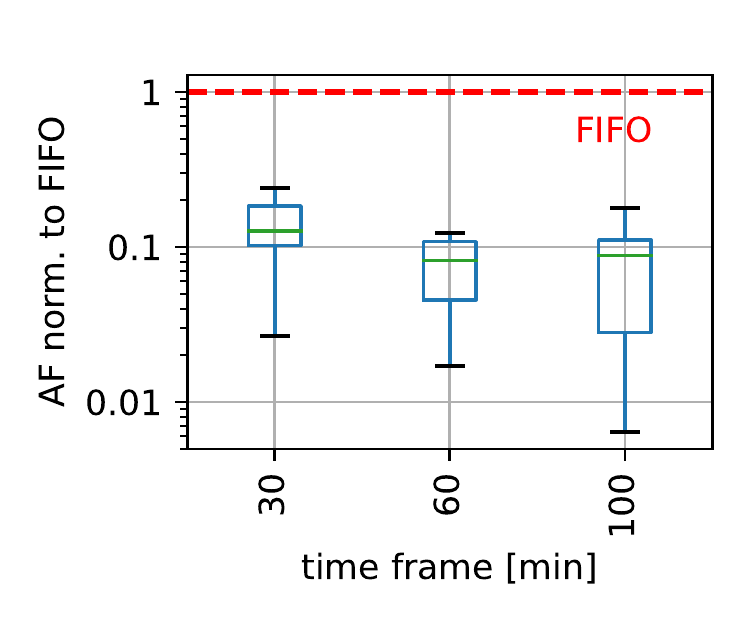}}}%
    \subfloat[SERPT-RE-LIM + AS]{{\includegraphics[width=0.25\textwidth,trim={0 2mm 0 5mm},clip]{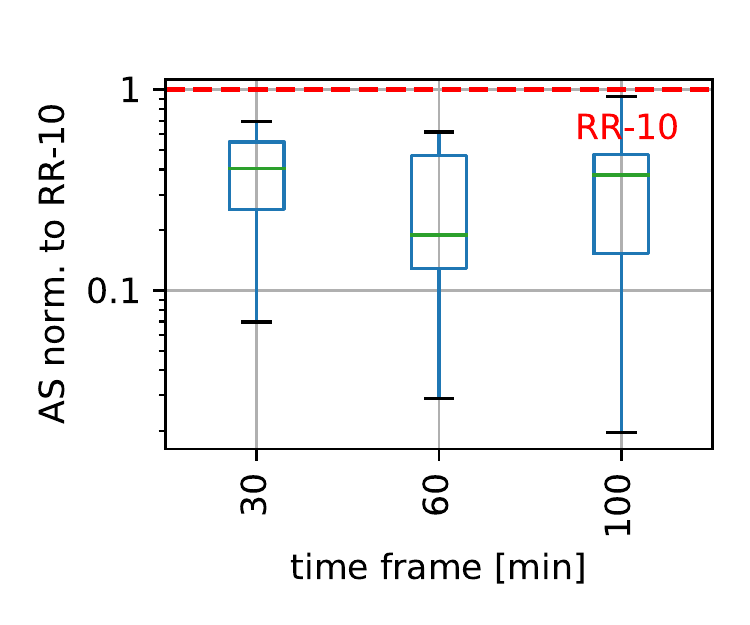}}}%
      
    \caption{Comparison of the algorithms with different time frame durations. 20 processors and 90\% load.}
	\label{fig:cmp_duration}
\end{figure*}

Finally, we verified whether 30-minute time frame is reasonable by providing results that can be extrapolated on larger time windows.
In Fig.~\ref{fig:cmp_duration} we present representative results on samples generated with different time frame durations.
All the results are comparable, which indicates that our algorithms provide similar results for longer time spans.

\section{Related Work}

\subsection{Serverless and FaaS scheduling}

As serverless model is flexible and there are various scenarios of its application, there exists a wide range of different aspects of this model that have to be addressed~\cite{Baldini2017,jonas2019cloud}.
An important issue connected with FaaS scheduling is the cold start---an additional overhead required to prepare execution environment at the time of the first invocation.
Particle~\cite{thomas2020particle} identify network provisioning as an important factor influencing startup time in platforms using containers.
To address this issue, proposed solution decouples creation of network from container creation process and uses pool of ephemeral IPs.

\cite{8752939} proposes decoupling common dependencies from function code---\emph{packages} (e.g. libraries)---cache them directly on worker nodes and schedule invocations taking into account package availability.
Such approach leads to the reduction of startup time as large libraries do not have to be downloaded multiple times.
Further analysis of performance challenges connected with FaaS is presented in~\cite{10.1145/3185768.3186308}.

Fifer~\cite{10.1145/3423211.3425683} considers execution of functions chains (i.e. invocation patterns where one function invokes a next one) and take into account the effects of cold starts.
Fifer implements its prediction model for incoming invocations using LSTM.

Massive parallel invocations of cloud functions can easily lead to an exhaustion of cluster resources.
In \cite{dukic2020photons}, inefficient memory usage (e.g. redundant data runtime, libraries) is addressed by introducing context sharing between multiple invocations.

\subsection{Scheduling with fixed processing times}

A number of theoretical research papers on scheduling with release dates was focused on fixed job execution times. As the release date of each job, $r_i$, is schedule-independent, the total flow-time objective, $\sum(C_i - r_i)$, becomes equivalent to the total completion time, \csum{}. It was shown that although the $\text{P}m||\csum$ problem is polynomially-solvable \cite{BrunoCoffman1974}, the $\text{P}m|\text{pmtn}, r_j|\csum$ \cite{Bellenguez2015,Baptiste2007} and even $1|r_i|\csum$ \cite{LenstraRinnooy-Kan1977} problems are strongly \np-Hard. For these general problems, some online algorithms were analyzed. The modified version of the delayed SPT strategy \cite{Hoogeveen1996}) was shown to provide a $2$-competitive ratio for the $\text{P}m|\text{on-line}, r_j|\csum$ problem \cite{Liu2009}. The same competitive ratio can be achieved for the $\text{P}m|\text{on-line}, \text{pmtn}, r_j|\cwsum$ problem \cite{Megow2004}, but only a $2.62$-competitive algorithm is known for the non-preemptive case \cite{Correa2005}. These results were extended in \cite{Leonardi2007} where it was proven that, for $P = p_{\text{max}}/p_{\text{min}}$ being the ratio of the maximum to the minimum job execution time, SRPT is a $O(\log(\min\{\frac nm, P\}))$-approximation offline algorithm for the $\text{P}m|\text{pmtn},r_j|\csum$ problem.

In case of the total stretch (\ssum{}) objective, the results are even less promising. It was shown that the SRPT strategy is $14$-competitive for the $\text{P}m|\text{on-line}, \text{pmtn}, r_j|\ssum$ problem \cite{Muthukrishnan1999}. In \cite{Bender2004}, it was shown that there exists a strategy with a constant-factor competitive ratio for a uniprocessor machine even is the processing times are known only to within a constant factor of accuracy. There was also shown a PTAS for the offline version of the $1|\text{pmtn},r_j|\ssum$ problem \cite{Bender2004}.

The maximum-defined objectives, such as $\max\{C_i - r_i\}$ (\fmax) and $\max\{(C_i - r_i)/p_i\}$ (\smax), were also discussed. A number of results were shown in \cite{Bender1998}. In particular, it was proven that for the $\text{P}m|\text{on-line}, r_i|\fmax$ problem, FIFO is a $(3 - 2/m)$-competitive strategy, and this bound is tight. It is also known that the offline version of the $1|r_i|\smax$ problem cannot be approximated in polynomial time to within a factor of $O(n^{1-\varepsilon})$, unless $\p=\np$ \cite{Bender1998}. Finally, it is shown that for $P = p_{\text{max}}/p_{\text{min}}$, every online algorithm for the $\text{P}m|\text{on-line}, \text{pmtn}, r_i|\smax$ problem is $\Omega(P^{1/3})$-competitive for three or more job sizes.

\subsection{Stochastic approaches to scheduling}

The results for fixed job processing times provide us lower bounds for the expected performance in case of the corresponding stochastic problems. In practice, stochastic problems are more complex. For example, it was shown that the performance guarantee for the $\text{P}m|\text{on-line}, r_j|\mathbb{E}[\cwsum]$ problem can be upper-bounded by $\frac{5+\sqrt{5}}{2}-1/(2m)$, if the expected remaining processing time of any job is a function decreasing in time \cite{Megow2005}, compared to a 2.62-competitive ratio in case of a non-stochastic variant \cite{Correa2005}. A good review of research papers considering stochastic scheduling problems, mostly non-preemptive, can be found in \cite{Megow2006,Vredeveld2011}. It can be observed that most results are related to the \cwsum{} objective.

To the best of our knowledge, there is only a limited number of papers on preemptive stochastic scheduling (e.g. \cite{Megow2006A,Megow2014}). However, performance guarantees are shown only for specific distributions of processing times (i.e. discrete ones). Moreover, we found no theoretical papers on stochastic scheduling with average or maximum stretch as performance metrics. 

\section{Conclusions}
This paper was driven by real-world data provided in the Azure Function Trace. We studied various non-clairvoyant, online scheduling strategies for a single node in a large FaaS cluster. Our aim was to improve performance measured with metrics related to response time or stretch of the function invocations. To estimate the values of the expected processing time or the expected remaining processing time of an invocation, we took advantage of the fact that the same function is usually invoked multiple times. This way, we were able to adapt SEPT and SERPT strategies with no significant increase in the consumption of memory or computational power. For our newly-introduced fair metrics, the function-aggregated stretch and flow time, we proposed two new heuristics, called Fair Choice. There, decisions are made based on an additional estimation of the expected number of function calls in the next monitoring interval.

Compared to round-robin and FIFO baselines, in the base case of our simulations, our proposed SEPT and SERPT strategies reduce the average flow time by a factor of 1.4 (preemptive) to 6 (non-preemptive); and the average stretch by a factor of 2.6 (preemptive) to 50 (non-preemptive). Gains over FIFO and round-robin increase with increased pressure of the workload on the system: with the lower number of processors and the higher average load. For the fair, function-aggregated metrics, our newly introduced Fair Choice strategies clearly outperform other implementable algorithms when measuring the flow time (while the gain is smaller for stretch).

SEPT, SERPT and Fair Choice can be easily implemented in the node-level component of the FaaS scheduling stack (e.g. the Invoker module in OpenWhisk). We expect the actual improvement compared to round-robin and FIFO baselines in practical cases, as our simulations were based on real-life data.

\section*{Acknowledgements}

This research is supported by a Polish National Science Center grant Opus (UMO-2017/25/B/ST6/00116).

\balance

\bibliographystyle{IEEEtran}
\bibliography{article}

\end{document}